\documentclass[format=acmsmall, review=false]{acmart}

\usepackage{acm-ec-18-copy}

\usepackage{booktabs} % For formal tables

\usepackage{amsmath}
\usepackage{array,xspace,multirow,hhline,graphicx,xcolor,tikz,colortbl,tabularx,amsmath,amssymb,amsfonts}
\usetikzlibrary{shapes}
\usetikzlibrary{arrows}
\usepackage{amsthm}

\usepackage{graphicx}
\usepackage{caption}

\usepackage{wrapfig}
\newcommand\eat[1]{}
\usepackage{pgfplots}
\usepackage{bbm}

\usepackage{array,xspace,hhline,tikz,colortbl,tabularx,fixltx2e,amsmath,amssymb,amsfonts,amsthm}
\usepackage{verbatim,ifthen}
\usepackage{pifont}
\usepackage{ifthen}
\usepackage{pgflibraryshapes}

\usepackage{nicefrac}
\usetikzlibrary{arrows}
\usepackage{ltxtable}
\usepackage{longtable}

	\usepackage{varioref}

\tikzset{
  jumpdot/.style={mark=*,solid},
  excl/.append style={jumpdot,fill=white},
  incl/.append style={jumpdot,fill=black},
  rexcl/.append style={jumpdot,color=red,fill=white},
  rincl/.append style={jumpdot,fill=black,color=red},
}
\usepackage[ruled]{algorithm2e} % For algorithms

\SetAlFnt{\small}
\SetAlCapFnt{\small}
\SetAlCapNameFnt{\small}
\SetAlCapHSkip{0pt}
\IncMargin{-\parindent}

	\newtheorem{remark}{Remark}

	\newcommand\blfootnote[1]{%
  \begingroup
  \renewcommand\thefootnote{}\footnote{#1}%
  \addtocounter{footnote}{-1}%
  \endgroup
}

\usepackage{mathtools}
\DeclarePairedDelimiter\ceil{\lceil}{\rceil}
\DeclarePairedDelimiter\floor{\lfloor}{\rfloor}
\newcommand{\instance}{$\langle \boldsymbol{x}, s, k\rangle$\xspace}

\newcommand{\stp}{DIC\xspace}

\begin{document}
% Title portion. Note the short title for running heads
\title[The Capacity Constrained Facility Location Problem]{The Capacity Constrained Facility Location problem\footnote{This paper was previously circulated under the title ``Mechanism Design without Money for Common Goods."} }

\author{Haris Aziz}
%\orcid{1234-5678-9012-3456}
\affiliation{%
  \institution{UNSW Sydney and Data61 CSIRO}
  \streetaddress{Data61, CSIRO, and UNSW, Sydney, NSW 2052, Australia}
  \city{Sydney}
  \state{NSW}
  \postcode{NSW 2052}
  \country{Australia}
  }
\author{Hau Chan}
\affiliation{%
  \institution{University of Nebraska-Lincoln}
  \department{Department of Computer Science and Engineering}
  \city{Lincoln}
  \state{NE}
  \postcode{68588}
  \country{USA}
}
\author{Barton E. Lee}
%\orcid{1234-5678-9012-3456}
\affiliation{%
  \institution{UNSW Sydney and Data61 CSIRO}
  \streetaddress{Data61, CSIRO, and UNSW, Sydney, NSW 2052, Australia}
  \city{Sydney}
  \state{NSW}
  \postcode{NSW 2052}
  \country{Australia}
  }
  \author{David C. Parkes}
  \affiliation{%
    \institution{Harvard University}
    \department{John A. Paulson School of Engineering and Applied Sciences, }
    \city{Cambridge}
    \state{MA}
    \postcode{02138}
    \country{USA}
  }

  \blfootnote{*This paper was previously circulated under the title ``Mechanism Design without Money for Common Goods."}
\blfootnote{Authors' email addresses: Haris Aziz: {\color{blue}\url{haris.aziz@unsw.edu.au}};  Hau Chan: {\color{blue}\url{hchan3@unl.edu}}, Barton E. Lee: {\color{blue}\url{barton.e.lee@gmail.com}}, David C. Parkes: {\color{blue}\url{parkes@eecs.harvard.edu}}.}

  %\author[1]{Haris Aziz}
%\author[2]{Hau Chan}
%\author[1]{Barton E. Lee}
%\author[4]{David C. Parkes}

%	\affil[1]{Data61, CSIRO, and UNSW, Sydney, NSW 2052, Australia.}
%		\affil[2]{Department of Computer Science and Engineering, University of Nebraska-Lincoln, 1400 R Street, Lincoln, NE 68588, USA.}
%%	\affil[3]{Data61 and the School of Economics, UNSW, Sydney, Australia.}
%	\affil[4]{John A. Paulson School of Engineering and Applied Sciences, Harvard University, 33 Oxford Street, Maxwell
%Dworkin 242, Cambridge, MA 02138, USA.}

	%\institute{%
%\email{haris.aziz@unsw.edu.au, hchan3@unl.edu, barton.e.lee@gmail.com, parkes@eecs.harvard.edu }
% }

% note that the abstract must come before \maketitle
\begin{abstract}

We initiate the study of the capacity constrained facility location problem from a mechanism design perspective. The capacity constrained setting leads to a new strategic environment where a facility serves a subset of the population, which is endogenously determined by the ex-post Nash equilibrium of an induced subgame and is not directly controlled by the mechanism designer. Our focus is on mechanisms that are ex-post dominant-strategy incentive compatible (DIC) at the reporting stage. We provide a complete characterization of DIC mechanisms via the family of Generalized Median Mechanisms (GMMs). In general, the social welfare optimal mechanism is not DIC. Adopting the worst-case approximation measure, we attain tight lower bounds on the approximation ratio of any DIC mechanism. The well-known median mechanism is shown to be optimal among the family of DIC mechanisms for certain capacity ranges. Surprisingly, the framework we introduce provides a new characterization for the family of GMMs, and is responsive to gaps in the current social choice literature highlighted by Border and Jordan (1983) and Barbar{\`a}, Mass{\'o} and Serizawa (1998).
\end{abstract}

% note: this command has been disabled to remove the ACM copyright block. Sorry...
%\thanks{This work is supported by the National Science Foundation,
%  under grant CNS-0435060, grant CCR-0325197 and grant EN-CS-0329609.}

\maketitle

\section{Introduction}

%{\color{red} Warning: Work in progress}\\
%, or the position of a political candidate on the left-right political spectrum.

A common economic problem is deciding where a public facility should be located to service a population of agents with heterogenous preferences. For example, a government needs to decide the location of a public hospital, or library. More abstractly, the `location'  may represent a type or quality of a service.  For example, a government may have a fixed hospital location but must decide on the type of service the hospital will specialize, and, in particular, whether the service will be targeted to those suffering from acute, moderate, or mild severity of a certain illness. In such problems, participants may benefit by misreporting their preferences, and this can be  problematic for a decision maker trying to find a socially optimal solution. This leads to the mechanism design problem of providing optimal, or approximately optimal, solutions while also being \emph{strategyproof}, i.e., no agent can profit from misreporting their preferences regardless of what others report.\footnote{We focus on the `mechanism design without money' problem where the use of money is assumed to not be permitted. This is a natural assumption for environments where the use of money is considered  unlawful (e.g., organ donations) or  unethical (e.g., political decision making, or locating a public good).} We call this the \emph{facility location problem}.

A large literature has studied the facility location problem under the assumption that the facility does not face capacity constraints. When the facility is not capacity constrained, all agents can benefit from the facility and hence it is modeled as a \emph{public good}.\footnote{A public good is non-rivalrous and non-excludable.}  Under this assumption, the mechanism design problem is explored in several classic papers~\cite{Blac48,Gibb73,Gibb77a,Satt75,Moul80,BoJo83}, and more recently in  algorithmic mechanism design~\cite{PrTe13,NiRo01,FFG16}.

To the best of our knowledge, an unexplored setting for the mechanism design problem is where the public facility is capacity constrained.\footnote{There is a distinct setting sometimes referred to as the `constrained facility location' problems~\cite{SuBo15} where the feasible locations for the facility are constrained. The algorithmic problem, of locating multiple capacity constrained facilities when agents are not strategic, has also been studied~\cite{BrCh89,KPTW01,Vygen05:Approximation}.} Capacity constraints limit the number of agents who can benefit from the facility's services. Such constraints are ubiquitous in practice: a hospital is capacity constrained by the number of beds and doctors, and a library may have limited seating.  When present, capacity constraints introduce a particular form of \emph{rivalry} to the facility, since once the facility reaches its capacity limit additional agents are prevented from using, and hence benefiting, from the facility.

A number of new strategic challenges arise for the mechanism designer when the public facility is capacity constrained but is still non-excludable. For example, when the mechanism designer chooses a location for the facility, we cannot stipulate which agents  will be served, instead these decisions are made by
participants, through strategic interactions once the facility has been located. That is, the ex-post Nash equilibrium of a subgame induced by the facility location determines the agents who ultimately benefit from the facility and those who do not. This introduces a technical challenge, because it leads  agents to have interdependent utilities, since the utility for a particular location depends on who else will use the location (and in turn on their preferences). Furthermore, the reporting game is made in anticipation of the extensive-form game and ex-post Nash equilibrium, and  the designer must consider mechanisms that are strategyproof in this broader game-theoretic context.

In this paper, we initiate the study of the capacity constrained facility location problem from the viewpoint of mechanism design. In our model, $n$ agents are located in the $[0,1]$ interval, and there is a single facility to be located, this facility is able to service at most $k$ agents, where $k$ is some positive integer. When $k\ge n$ the capacity constraint is of no effect, and the capacity constrained facility location problem is equivalent to the classic problem. Agent locations are privately known, and, given a facility location, the ex-post Nash equilibrium of an induced subgame determines which agents are served.  The mechanism designer's problem is to design mechanisms that are strategyproof and maximize social welfare. In our model, we take strategyproof to  mean ex-post dominant-strategy incentive compatible (\stp) at the reporting stage. That is, conditional on the ex-post Nash equilibrium being attained in the induced subgame, an agent never benefits ex-post from misreporting their location to the mechanism regardless of what other agents report, and regardless of other agents' true locations. For ease of exposition, a mechanism that is \stp at the reporting stage will simply be said to be \stp. Unlike the classic facility location problem where the facility is not capacity constrained, the social welfare optimal mechanism is not \stp  except when the capacity constraint is trivial, i.e., $k=1$ or $n$. As a result, we follow the approach of Procaccia and Tennenholtz~\cite{PrTe13} and consider the approximate mechanism design problem. We adopt the worst-case approximation measure for social welfare, and ask what is the best approximation achievable  with a \stp mechanisms and how does this vary as a function of the capacity constraint?

The literature studying the facility location problem without capacity constraints, or simply $k=n$, provides a number of important results. Gibbard-Satterthwaite~\cite{Gibb73,Satt75} showed a powerful impossibility result: when agents can have unrestricted preferences there need not exist any strategyproof mechanism. As a result, more recent works typically restrict agent preferences' over the location of the facility to be single-peaked and sometimes in addition symmetric.\footnote{A single-peaked preference is symmetric if equidistant locations on either side of the ideal, or `peak', location are always equally preferred.} We focus on the case where, conditional on the agent being served, the agent has preferences that are both single-peaked and symmetric. When the objective of the mechanism designer is to maximize social welfare, i.e., utilitarian welfare, the standard median mechanism is both strategyproof and social welfare optimal~\cite{Blac48}. More generally, a goal of the social choice literature has been to characterize the complete family of strategyproof mechanisms. Closest to our setting, Border and Jordan~\cite{BoJo83} provide a partial characterization of strategyproof mechanisms  via the family of Generalized Median Mechanisms (GMMs). Border and Jordan show that a mechanism is strategyproof and unanimity respecting\footnote{Unanimity respecting simply means that if there is a unanimously most preferred facility location then the mechanism must locate the facility at this location.} if and only if it is a GMM, and that the family of GMMs is strictly smaller than the complete family of strategyproof mechanisms.\footnote{We note that in a slightly different setting, where the single-peaked preferences are possibly asymmetric, GMMs provide a complete characterization of strategyproof and `peak only' mechanisms (Proposition 3 of Moulin~\cite{Moul80}). Example~\ref{Example: DIC hard} in the present paper provides an example of a mechanism that is strategyproof in the Border and Jordan~\cite{BoJo83} setting but not the Moulin~\cite{Moul80} setting.} This has left a gap in the literature to characterize the complete family of strategyproof mechanisms and understand the difference in strategyproof mechanisms that are GMMs and those that are not. Figure~\ref{Figure: Summary of border jordan 1983} schematically illustrates this gap.\\

\textbf{Our Contributions:} We introduce a new mechanism problem, the capacity constrained facility location problem. This problem is a natural variant of the classic facility problem where the facility is assumed to face capacity constraints. A conceptual contribution  is to formalize the effect of capacity constraints when the facility is non-excludable but cannot service all agents. We do this by defining an extensive-form game involving the mechanism designer and agents. First, agents report their preferences to the designer, and then the facility is located by the mechanism. Once the facility is located a subgame is induced where agents strategically choose whether or not to attempt to be served by the facility. The ex-post Nash equilibrium determines which  agents are served by the facility and which are not.
% That is, the set of agents served by the facility is endogenously determined. 
%Formally defining this subgame allows the mechanism designer's problem, at the reporting/first stage, to be well-defined. In particular,
We seek mechanisms that are strategyproof in this broader game-theoretic context, i.e.,  ex-post dominant-strategy incentive compatible; that is, conditional on the ex-post Nash equilibrium being achieved in the subgame, no agent can benefit from misreporting their location regardless of what other agents report and regardless of other agents' true locations.

% for a capacity constrained facility. 
%
%A conceptual and novel contribution is the formalization of the model which requires the consideration of an extensive-form game that endogenously derives the subset of agents served by the facility.

 Our main theoretical contribution is a complete characterization of \stp mechanisms for the capacity constrained facility location problem. We show that a mechanism is \stp if and only if it belongs to the established family of mechanisms called the Generalized Median Mechanisms (GMMs), which appear in Moulin~\cite{Moul80} and Border and Jordan~\cite{BoJo83}. Thus, the framework we introduce  surprisingly provides a new characterization of GMMs. This result contributes to a novel perspective to a ``major open question" (Barbar{\`a}, Mass{\'o} and Serizawa~\citep{BMS98}) posed in Border and Jordan~\cite{BoJo83} (further discussion is provided in Section~\ref{Section: related lit}).

 We also provide algorithmic results and study the performance of \stp mechanisms in optimizing social welfare. We adopt the  worst-case approximation measure, and provide a lower bound on the approximation ratio of any \stp mechanism. We show that at best the approximation ratio of a \stp mechanism is $2\frac{k}{k+1}$ when $k\le \ceil{(n-1)/2}$, and $\max\{\frac{n-1}{k+1}, 1\}$ otherwise. Interestingly, this lower bound is achieved by the standard median mechanism (which is also \stp) when $k\le  \ceil{(n-1)/2}$ or $k=n$, and hence the median mechanism is optimal among all \stp mechanisms in those ranges. Figure~\ref{figure: illustration DIC lower boundxxy} illustrates these approximation results.

 Finally, we consider an extension of our framework where the
 mechanism designer can also restrict access to the facility, and
 hence dictate which agents are served. This extension is relevant to
 settings where the designer can issue permits, and prevent certain
 agents from accessing the facility. Under an anonymity assumption, we show that no mechanism that both locates the facility
 and stipulates which agents can be served is \stp.

\begin{figure}[H]
\centering 
\begin{tikzpicture}[scale=1.05,
  declare function={
    func(\x)= (\x <= 50) * (2*(\x)/(\x+1))   +
                   (\x>50) * (99/(\x+1)) % changed cutoff from 49 to 50 to make picture nicer
   ;
    funcy(\x)= (\x <= 50) * (2*(\x)/(\x+1))   +
                          (\x>50) *(2*(\x)/(\x+1)) %+
              %     (\x>50) * (  (2*(\x)/(\x+1))> (1+ 2*(100-\x+1)/(3*\x-2(100)-2) )    )*(2*(\x)/(\x+1)) +
                 %  (\x>50) * ((2*(\x)/(\x+1))<= (1+ 2*(100-\x+1)/(3*\x-2(100)-2)))*( 1+ 2*(100-\x+1)/(3*\x-2*(100)-2))
   ;
       funcz(\x)= (\x <= 81) * (2*(\x)/(\x+1))   +
                       (\x>81) *( 1+ 2*(100-\x+1)/(3*\x-2*(100)-2))
                          ;
  }
]
\begin{axis}[
  axis x line=middle, axis y line=middle,
  ymin=0, ymax=3, ytick={0,1,2,3}, ylabel={$\alpha$-approximation},
  xmin=0, xmax=100, xlabel=$k$,
  xtick={0,1,25,50,75,100},
xticklabels={$0$, 1,$n/4$,$n/2$, $3n/4$,$n$},
%xmajorgrids,
  domain=0:100,samples=201, % added
]

%\draw[fill=gray!50!white] plot[smooth,samples=100,domain=5:10] (\x,{ln(\x)}) -- 
   % plot[smooth,samples=100,domain=5:10] (\x,{1});

\addplot [blue,ultra thick, domain=1:100] {func(x)};

%\addplot [red,ultra thick, domain=1:100] {funcy(x)};
\addplot [red,ultra thick,dotted, domain=1:100] {funcz(x)};

\legend{\stp lower bound, median mech. upper bound}
\end{axis}
\end{tikzpicture} 
\captionsetup{justification=centering,margin=2cm}
\caption{Worst-case approximation ratio as a function of the capacity constraint, $k$.}
\label{figure: illustration DIC lower boundxxy}
\end{figure}
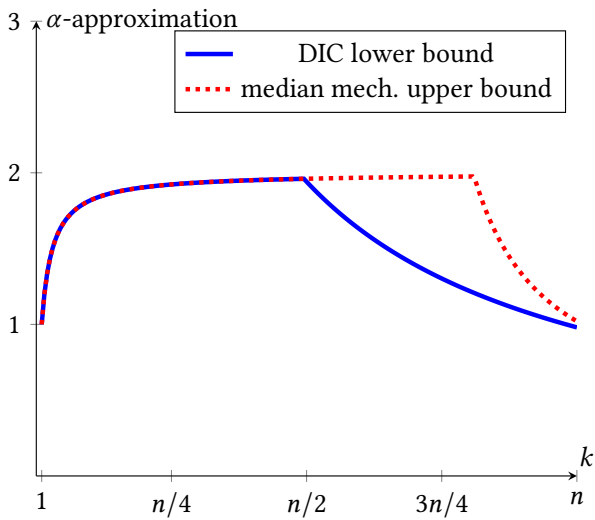

\textbf{Outline:} Section~\ref{Section: related lit} provides a brief literature review. Section~\ref{sec:model} presents our model and formalizes the objective of the mechanism designer, Section~\ref{Section: characterization} then presents our key characterization result of \stp mechanisms. Section~\ref{Section: approxim of dic} explores the performance, i.e., approximation results, of \stp mechanisms. Section~\ref{Section: Excludable} considers an extension of our framework where the mechanism designer is able to dictate which agents are served by the facility. Lastly, we conclude with a discussion in Section~\ref{section: discussion and conlc}.

\section{Related literature}\label{Section: related lit}

A number of papers have considered related mechanism problems where the use of money is not permitted~\citep{AsRo11,AbSo03,PrTe13,Moul80,BoJo83,Gibb73,Satt75,SuBo15,MLY+16}. Most closely related to our paper is~\cite{PrTe13}, where agents with single-peaked preferences are located along the real line and the problem of locating a (non-capacity constrained) public facility is studied with the goal of minimizing two distinct objective functions; the total social cost and the maximum social cost. This problem is often referred to as a single facility location problem, or single facility location game.\footnote{We do not review a large segment of computer science and operations research literature on facility location problems that assumes complete information and hence does not require a mechanism design approach to overcome strategic tensions (for a survey see~\cite{BrCh89}). Furthermore, this literature, when incorporating capacity constraints, typically focuses on the problem of locating multiple capacity constrained facilities that have sufficient capacity to service all agents~\cite{Cygan12:LP,Charikar02:Constant,KPTW01,Vygen05:Approximation}. Instead we review the subset of literature  that assumes strategic agents and takes a mechanism design approach.}  In this paper, we focus on minimizing the first objective function in the new environment where the facility is capacity constrained. In contrast to the setting studied by~\cite{PrTe13}, agents have interdependent utilities, in our model, due to the capacity constraints of the facility and the induced subgame. Accordingly, the mechanism design problem requires consideration of a broader game-theoretic environment where agents face an extensive-form game when reporting preferences.

%and are only weakly single-peaked.\footnote{Informally speaking, weakly single-peaked preferences means that alternatives further away from the `peak' are weakly less preferred but not necessarily strictly less preferred, as per the standard single-peaked definition. } 
%
%Furthermore, due to capacity constraints we are forced to consider a broader game-theoretic setting where agents face an extensive-form game when reporting preferences.

%Thus, the results of~\cite{Moul80,BoJo83} are no longer immediately applicable, and a range of distinct results are attained. Furthermore, the interdependence of utilities means that the \truth property may be strictly weaker than the \stp property; in the setting considered by~\cite{PrTe13} these two properties are equivalent. 
% \Hau{Do we show this? or did they show that ex post IC is different from DIC.}
%   In multi-dimensional space, \cite{BoJo83} shows that a mechanism is \stp and unanimity respecting if and only if the mechanism is a GMM. 
 
 Another large body of literature has been concerned with characterizing \stp mechanisms for the unconstrained facility location problem. The key pioneering works in this area are by Moulin~\cite{Moul80}, and Border and Jordan~\cite{BoJo83}. In one-dimensional space and for symmetric and single-peaked preferences, Border and Jordan~\cite{BoJo83} characterize a general class of \stp mechanisms which have become to be known as \emph{generalized median mechanisms} (GMM), and in addition show that when the property of unanimity is enforced every \stp mechanism is a GMM.\footnote{Border and Jordan~\cite{BoJo83} also consider the problem in higher dimensions.} These results differ slightly from the characterization results of Moulin~\cite{Moul80} since the setting studied in~\cite{Moul80} does not restrict the single-peaked preferences to be symmetric.  Characterizing \stp but non-unanimity respecting mechanisms was posed as an open problem; as stated by Border and Jordan in~\cite{BoJo83} ``\emph{[the characterization] leaves several open problems. The most obvious question is: what happens if the unanimity assumption is dropped?}" Characterizations however, have remained elusive and it has become known as a ``\emph{major open question}"~\citep{BMS98} with only partial progress towards a resolution~\citep{Chin97,BMS98,PPS+97,Weym11}. In this paper we focus on the one-dimensional case where open questions still remain; in particular, the results of~\cite{BoJo83} in one-dimensional space leaves two gaps: 
 \begin{enumerate}
 \item there exist non-unanimity respecting \stp mechanisms that are not GMM, and 
 \item there exist \stp mechanisms that are GMMs but do not respect unanimity. 
 \end{enumerate} 
Our characterization of \stp mechanisms via the family of GMM, although considered in a different setting where the facility is capacity constrained, applies more generally to mechanisms that are not unanimity respecting. Hence, we contribute a novel perspective to these gaps in characterization, showing that a mechanism is \stp for all possible capacity constraint $k\le n$ if and only if it is a GMM. This means that any mechanism in gap (1) is not \stp when the facility is capacity constrained with $k<n$. Furthermore, the unanimity property is sufficient to ensure that a mechanism that is \stp in the non-capacity constrained setting remains \stp when capacity constraints are present.

% which we believe will be a useful contribution to the more general multi-dimensional open problem. We characterize precisely the complete set of GMMs as the family of mechanisms which are \stp when the public good is capacity constrained (or threshold-rivalrous). Thus, we characterize the mechanisms in the gap (ii) identified above. One interpretation of our result is that GMMs are equivalent to the family of \emph{capacity-robust-\stp} mechanisms since for any capacity constraint (including no capacity constraint) these mechanisms are \stp. Thus, our results also show that the mechanisms in gap (i) are not \emph{capacity-robust-\stp}.

\section{Model, Basic Properties, and Definitions} \label{sec:model}

\textbf{Model:} Let $N=\{1,\ldots, n\}$ be a finite set of $n$ agents and let $X=[0,1]$ be the domain of agent locations. Each agent $i\in N$ has a location $x_i\in X$, which is privately known,  the profile of agent locations is denoted by $\boldsymbol{x}=(x_1,  \ x_2, \  \ldots \ , \ x_n)$. The profile of all agent except some agent $i\in N$ is denoted by $\boldsymbol{x}_{-i}=(x_1, \ x_2,\  \ldots \ , x_{i-1}, \ x_{i+1}, \ \ldots \ , \ x_n)$. There is a single facility to be located in $X$ by some mechanism. A \emph{mechanism} is a function $M:  \prod_{i\in N} X\rightarrow X$ mapping a profile of locations to a single location.\footnote{We restrict our attention to deterministic mechanisms.} We denote the mechanism's output, or facility location, by $s\in X$.

The facility faces a capacity constraint $k \ : \ k \le n$, which provides a limit on the number of agents that can be served.  A \emph{served} agent attains utility $u_i=1-d(s, x_i)\ge 0$, where $d(\cdot, \ \cdot)$ denotes the Euclidean metric;  an \emph{unserved} agent attains zero utility, $u_i=0$.\footnote{Our characterization results (Section~\ref{Section: characterization}) do not rely on this specific utility function -- we only require that agents weakly prefer to be served than not, and conditional on being served the agent's utility is symmetric and (strictly) single-peaked. However, our approximation results do rely on the choice of utility function.} The set of agents served by the facility's limited capacity is not directly controlled by the mechanism, since the facility is assumed to be non-excludable.\footnote{In Section~\ref{Section: Excludable} we weaken this assumption and consider the problem when the facility can be made excludable.} Instead, this is determined by the equilibrium outcome of a subgame induced by the mechanism's choice of facility location. 

%The \emph{social welfare} refers to the total sum of agent utilities. 

Given an instance $\langle \boldsymbol{x}, s, k\rangle$, we assume that the set of agents served by the facility is determined via the ex-post Nash equilibrium\footnote{That is, no agent has an incentive to unilaterally deviate, whatever the preferences of each agent.} of a subgame, $\Gamma_{\boldsymbol{x}}(s,k)$. The subgame $\Gamma_{\boldsymbol{x}}(s,k)$ is as follows. Each agent $i\in N$ chooses an action $a_i\in A= \{\emptyset, s\}$ of whether, or not, to travel from their location $x_i$ to the facility location $s$. Action $a_i=s$ denotes agent $i$'s choice to travel to the facility, and action $a_i=\emptyset$ denotes the agent's choice to not travel to the facility. We denote the profile of agent actions by $\boldsymbol{a}=(a_1, \ a_2, \ \ldots \ , a_n)$. An agent $i$ is \emph{served} by the facility if they travel to the facility, $a_i=s$, and strictly less than $k$ other agents travel to the facility, i.e., $|N(\boldsymbol{a},s)|\le k$ where $N(\boldsymbol{a},s):=\{i\in N\  : \ a_i=s\}$. If instead they travel to the facility and at least $k$ other agents also travel to the facility, i.e., $|N(\boldsymbol{a},s)|>k$, then a tie-breaking rule is used to determine which subset of $k$ agents in $N(\boldsymbol{a},s)$ are served. We assume a distance-based tie-breaking rule ($\triangleright$) whereby agent $i$ has higher priority than agent $j$, denoted $i\triangleright j$, if agent $i$ is closer to the facility than agent $j$, i.e., $d(s, x_i)<d(s, x_j)$; if agent $i$ and $j$ are equidistant, i.e., $d(s, x_i)=d(s, x_j)$, then we apply some deterministic tie-breaking rule.\footnote{This ensure that the binary relation $\triangleright$ is complete.} This distance-based tie-breaking rule can be motivated by a `first-come-first-serve' protocol when the location $s$ is geographical and agents physically travel to the facility to be served. If the location $s$ corresponds to a type, or quality, of service the `first-come-first-serve' protocol is analogous to a `best-fit' tie-breaking protocol that prioritizes agents according to how close the type of service being offered is to their true needs, i.e., $d(s,x_i)$. An agent $i$ with location $x_i$ attains utility $1-d(s, x_i)$ if $a_i=s$ and they are served, if $a_i=s$ and they are not served they attain utility $-d(s,x_i)$, and otherwise $a_i=\emptyset$ and agent $i$ attains zero utility.

Abusing terminology slightly, given a profile of locations $\boldsymbol{x}$ and facility location $s$, we will refer to $k$ highest priority agents with respect to the distance-based tie-breaking rule ($\triangleright$) as the \emph{$k$-closest} agents. We denote this set of agents by $N_k^*(\boldsymbol{x}, s)$. \\

\textbf{Basic properties of the model:} For any instance $\langle \boldsymbol{x}, s, k\rangle$, the subgame $\Gamma_{\boldsymbol{x}}(s, k)$ has an (essentially) unique equilibrium. There always exists an equilibrium where the $k$-closest agents, $N_k^*(\boldsymbol{x}, s)$, choose to travel to the facility and are served by the facility. In instances where one or more of the $k$-closest agents are indifferent between being served and not traveling to the facility, i.e., whenever $d(s, x_i)=1$ for some $i\in N_k^*(\boldsymbol{x}, s)$, multiple equilibria arise. For the purposes of this paper these equilibria are all `equivalent' since every agent attains the same utility in each of the equilibria. Proposition~\ref{Proposition: Basic prop 1} states this basic property. The proof is straightforward and left to the appendix for the interested reader. 

\begin{proposition}\label{Proposition: Basic prop 1}
For any instance $\langle \boldsymbol{x}, s, k\rangle$, there exists an equilibrium of the subgame $\Gamma_{\boldsymbol{x}}(s, k)$ and, furthermore, in every equilibrium agent $i\in N$ attains utility $1-d(s, x_i)$ if $i\in N_k^*(\boldsymbol{x}, s)$, and otherwise, attains zero utility. 
\end{proposition}

Given Proposition~\ref{Proposition: Basic prop 1}, we can denote agent $i$'s ex-post equilibrium utility from the facility location $s$ by simply $u_i^*(s, \boldsymbol{x}, k)$. A useful observation is that the agent's ex-post utilities are (weakly) single-peaked, this result is stated in Proposition~\ref{Proposition: Basic prop 2}. The proof is straightforward and left to the appendix for the interested reader. Intuitively, the result holds because under the distance-based priority ($\triangleright$) an agent's priority only (weakly) improves when the facility moves from a location $s<x_i$ to a new location $s' \ : \ s<s'\le x_i$ (similarly for $s>x_i$).

\begin{proposition}\label{Proposition: Basic prop 2}
For any agent $i\in N$ and any pair of instances \instance and $\langle \boldsymbol{x}, s', k\rangle$, if $s<s'\le x_i$  or $x_i\le s'<s$ then $u_i^*(s, \boldsymbol{x}, k)\le u_i^*(s', \boldsymbol{x}, k)$.
\end{proposition}

In this paper we are interested in `strategyproof' mechanisms where agents do not have an incentive to misreport their location.  In particular, we use the ex-post \emph{Dominant-strategy Incentive Compatible} (DIC) concept of strategyproofness. That is, a mechanism $M$ is  \stp if for every agent $i\in N$ 
$$u_i^*\Big(M(x_i, \hat{\boldsymbol{x}}_{-i}), \boldsymbol{x}, k \Big)\ge u_i^*\Big(M(x_i', \hat{\boldsymbol{x}}_{-i}), \boldsymbol{x}, k \Big)$$
for every $x_i'$, for every $\hat{\boldsymbol{x}}_{-i}$, and for every $\boldsymbol{x}_{-i}$. Note that \stp implies that, conditional on the ex-post Nash equilibrium being achieved in the subgame $\Gamma_{\boldsymbol{x}}(s, k)$, the mechanism is dominant-strategy incentive compatible at the reporting stage. Formally speaking, the \stp definition depends on the capacity constraint $k$ however, abusing notation slightly, we omit the $k$ dependence as this will be clear from the context.\\

\textbf{Objective of the mechanism designer:} In this paper we are interested in DIC mechanisms that perform well with respect to \emph{social welfare}, i.e., the sum of agents' equilibrium utilities. As is now standard in the algorithmic mechanism design literature we measure the performance of a DIC mechanism by the worst-case approximation ratio. 

Given an instance \instance, denote the optimal social welfare by $\Pi^*(\boldsymbol{x}, k):=\max_{s\in X} \sum_{i=1}^n u_i^*(s, \boldsymbol{x}, k),$ and given a mechanism $M$ let $\Pi_M(\boldsymbol{x}, k)$ denote the social welfare attained by the mechanism, i.e., 
\begin{align*}
\Pi_M(\boldsymbol{x}, k)&:=\sum_{i=1}^n u_i^*(s, \boldsymbol{x}, k) &&\text{where $s=M(\boldsymbol{x})$.}
\end{align*}
The mechanism $M$ is an $\alpha$-approximation if
\begin{align}\label{Equation: approx 1}
\max_{\boldsymbol{x}\in \prod_{i=1}^n X}\Bigg\{\frac{\Pi^*(\boldsymbol{x},k)}{\Pi_M(\boldsymbol{x}, k)}\Bigg\}\le \alpha,
\end{align}
the LHS of (\ref{Equation: approx 1}) is referred to as the \emph{approximation ratio}. A mechanism (or family of mechanisms) is said to have a \emph{lower bound}, $\bar{\alpha}$, on the approximation ratio if 
\begin{align}\label{Equation: approx 2}
\bar{\alpha}\le \max_{\boldsymbol{x}\in \prod_{i=1}^n X}\Bigg\{\frac{\Pi^*(\boldsymbol{x},k)}{\Pi_M(\boldsymbol{x}, k)}\Bigg\}.\end{align}
We refer to a mechanism $M$ that attains the optimal social welfare for all instances \instance, and hence is an $\alpha=1$-approximation, as an \emph{optimal mechanism}.  Again, the optimal mechanism definition depends on the capacity constraint $k$ however, abusing notation, we will omit the $k$ dependence as this will be clear from the context. Note that the optimal mechanism need not, and in general will not, be \stp for a given $k$.\\

\begin{remark}\label{Remark: special case}
When $k=n$ our model reduces to the well-known facility location problem studied in~\cite{PrTe13,Moul80,Blac48}. Accordingly, this case ($k=n$) is fully resolved: the `median' mechanism which always locates the facility at the median reported location is both optimal and \stp.  However, the case for $k<n$ has not been studied before -- this is the focus of the present paper.
\end{remark}

To illustrate how the case where $k<n$ differs from the standard $k=n$ setting we provide an example. The example considers a mechanism that is \stp when $k=n$ but for any capacity constraint $k<n$ is not \stp.

\begin{example}\label{Example: DIC hard}
Let $M$ be the mechanism such that $M(\boldsymbol{x})=\text{arg}\min_{s\in \{ 1/4, \ 3/4\}} d(s, x_i)$ for some $i\in N$, tie-breaking in favor of $s=1/4$ if necessary. That is, the mechanism locates the facility at either location $1/4$ or $3/4$ depending on which is closest to agent $i$'s report.

First, notice that the mechanism $M$ is \stp when $k=n$. If $k=n$ then  every agent $i$ is always served by the facility and hence attains utility $1-d(s, x_i)$ for any facility location $s$. It is immediate that agent $i$ can never strictly benefit from misreporting their location. 

However, when $k<n$ the mechanism is not \stp. To see this, consider an instance where agent $i$ is located at $3/8$ and all other agents are located at $1/4$. When agent $i$ truthfully reports, the facility is located at $1/4$ and is not served -- leading to zero utility. On the other hand, misreporting to $x_i'\in (1/2,1]$ leads to the facility location $3/4$ and agent $i$ is the closest agent to the facility. In this case agent $i$ attains strictly higher utility equal to $1-d(3/4, 3/8)>0$. Thus, the mechanism is not \stp for any $k<n$.\hfill $\diamond$
\end{example}

%%%%%%%%%%%%%%%%%%%%%%%%%%%%%%%%%%%%%%%%%%%%%%%%

\subsection{A complete characterization of \stp mechanisms}\label{Section: characterization}

%%%%%%%%%%%%%%%%%%%%%%%%%%%%%%%%%%%%%%%%%%%%%%%%

We begin by defining a family of mechanisms called \emph{Generalized Median Mechanisms (GMM)}. This family was introduced by Border and Jordan~\cite{BoJo83} for the $k=n$ setting, and provides a partial characterization of \stp mechanisms. The main result of the present paper shows that GMMs provide a complete characterization of mechanisms that are (1) \stp for all $k\le n$, and (2) \stp for some $k<n$.

\begin{definition}{[Generalized Median Mechanism (GMM)]}
A mechanism $M$ is said to be a \emph{Generalized Median Mechanism} (GMM) if for each $S\subseteq N$ there are constants $a_S$, such that for all location profiles $\boldsymbol{x}$
\begin{align}\label{equation: GMM}
M(\boldsymbol{x})=\min_{S\subseteq N} \max\Big\{\max_{i\in S}\{x_i\}, a_S\Big\}.
\end{align}
\end{definition}

To build some intuition, we highlight some well-known GMM mechanisms:
\begin{enumerate}
\item The \emph{median mechanism}\footnote{A mechanism that always outputs the median of the reported location profile i.e., the $\floor{(n+1)/2}$-th smallest report.} is attained from (\ref{equation: GMM}) by setting $a_S=1$ for all subsets $S\subseteq N$ with $|S|<\floor{(n+1)/2}$ and $a_S=0$ otherwise. 
\item The \emph{$s$-constant mechanism}\footnote{A mechanism that always outputs the location $s$.} for some location $s\in X$, i.e., the mechanism that always outputs location $s$, is attained from (\ref{equation: GMM}) by setting $a_\emptyset=s$ and $a_S=1$ for all other (non-empty) subsets $S\subseteq N$. 
\item The \emph{agent $i$ dictatorship mechanism}\footnote{A mechanism that always outputs the location of agent $i$'s report.} is attained from (\ref{equation: GMM}) by setting $a_S=0$ for $S=\{i\}$ and $a_S=1$ for other subsets $S\subseteq N$.
\end{enumerate}
An example of a mechanism that is not a GMM is the dictatorial-style mechanism considered in Example~\ref{Example: DIC hard}.

The main result of the present paper is the following characterization: A mechanism $M$ is \stp for some $k< n$ if and only if $M$ is \stp for every $k\le n$ if and only if $M$ is a GMM. This result is stated in Theorem~\ref{Corollary: equivalence results}.

\begin{theorem}\label{Corollary: equivalence results}
Let $M$ be a mechanism. The following are equivalent: 
\begin{enumerate}
\item $M$ is a GMM,
\item $M$ is \stp for some $k<n$,
\item $M$ is \stp  for every $k\le n$.
\end{enumerate}
\end{theorem}

We present the proof via a series of propositions, and utilize a characterization of Border and Jordan~\citet{BoJo83}. Before presenting these propositions we illustrate the contribution of Theorem~\ref{Corollary: equivalence results}, benchmarked against the results of~\citet{BoJo83}: where GMM are shown to be a strict subset of \stp mechanisms when $k=n$. Below, in Figure~\ref{Figure: Summary of border jordan 1983}, we present the result of~\citet{BoJo83}. Figure~\ref{Figure: illustration of our results} illustrates our contribution. When considering the capacity constrained problem, with $k<n$, the family of \stp mechanisms coincides precisely with the GMM family.

\begin{figure}[H]
\centering 
\begin{minipage}{.5\textwidth}
\centering
   \begin{tikzpicture}[scale=0.8,
dot/.style = {circle, inner sep=0pt, minimum size=1mm, fill,
              node contents={}}
                        ]
\def\firstcircle{(0,0) coordinate (a) circle (2.75cm)}
\def\secondcircle{(0,-0.24) coordinate (b)  circle (2.3cm)}
\def\thirdcircle{(0,-0.9) coordinate (b)  circle (1.5cm)}
    \begin{scope}
%\clip \secondcircle;
\fill[gray] \secondcircle;
%\fill[cyan] \thirdcircle;

    \end{scope}
\draw \firstcircle; %node[dot,label=$n_{ij}$];
\draw \secondcircle;
 \node[circle] at (0, 2.35)  {DIC};
   \node[circle] at (0, 0.8)   {GMM};
%  \node[circle] at (0, -0.8)   {DIC ($k=n$) + unanimity};
%\node{(0,1)} {DIC mechanisms ($k=n$)};
%\node at (c -| b) {$t_j$}; 
    \end{tikzpicture}
    \caption{Setting where $k=n$~\cite{BoJo83}.}
    \label{Figure: Summary of border jordan 1983}
    \end{minipage}%
\begin{minipage}{.5\textwidth}
\centering 
   \begin{tikzpicture}[scale=0.8,
dot/.style = {circle, inner sep=0pt, minimum size=1mm, fill,
              node contents={}}
                        ]
\def\firstcircle{(0,0) coordinate (a) circle (2.75cm)}
\def\secondcircle{(0,-0.24) coordinate (b)  circle (2.3cm)}
\def\thirdcircle{(0,-0.9) coordinate (b)  circle (1.5cm)}
    \begin{scope}
%\clip \secondcircle;
\fill[white] \firstcircle;
\fill[gray] \secondcircle;

    \end{scope}
%\draw \firstcircle; %node[dot,label=$n_{ij}$];
\draw \secondcircle;
% \node[circle] at (0, 2.35)  {DIC ($k=n$)};
   \node[circle] at (0, 0.8)   {DIC $\equiv$ GMM};
  %\node[circle] at (0, -0.6)   {DIC ($k=n$) + unanimity};
     % \node[circle] at (0, -0.9)  {(equiv. };
%    \node[circle] at (0, -1.2)  {DIC ($k<n$) + unanimity)};
%\node{(0,1)} {DIC mechanisms ($k=n$)};
%\node at (c -| b) {$t_j$}; 
    \end{tikzpicture}
     \caption{Setting where $k<n$.}
    \label{Figure: illustration of our results}
    \end{minipage}
    \end{figure}

First, we present a result of~\citet{BoJo83} characterizing the family of GMMs via a property of the mechanism that they call `uncompromising'. Informally speaking, an uncompromising mechanisms means that an agent cannot influence the mechanism output in their favor by reporting extreme locations. The most obvious mechanism satisfying this property is the median mechanism. Formally, a mechanism $M$ is said to be \emph{uncompromising} if for every profile of locations $\boldsymbol{x}$, and each agent $i\in N$,  if $M(\boldsymbol{x})=s$ then
\begin{align}\label{eqaution: uncompromising 1}
x_i>s &\implies M(x_i',\boldsymbol{x}_{-i})=s\qquad \text{ for all } x_i'\ge s &&\text{ and,}\\
x_i<s &\implies M(x_i' , \boldsymbol{x}_{-i})=s \qquad \text{ for all } x_i'\le s.\label{eqaution: uncompromising 2}
\end{align}

\begin{lemma}[Border and Jordan~\cite{BoJo83}]\label{theorem: generalized median}
A mechanism $M$ is uncompromising if and only if it is a GMM. 
\end{lemma}
%Furthermore, every GMM is $k$-\stp for $k=n$.

Note that Lemma~\ref{theorem: generalized median}, although proved in the setting where $k=n$,  does not rely on any strategic properties of the mechanism and so applies more generally to our setting of interest where $k\le n$.

We now prove our first proposition towards the characterization result. Proposition~\ref{proposition: uncompromise connections} says that, every GMM is \stp for any $k\le n$. 

\begin{proposition}\label{proposition: uncompromise connections}
Every GMM is \stp for any $k\le n$.
\end{proposition}

\begin{proof}
%For this proof, we utilize the equivalence between GMMs and the uncompromising property (Theorem~\ref{theorem: generalized median}). 
Fix $k\le n$ and let $M$ be a GMM. For the sake of a contradiction suppose that $M$ not \stp. That is, for some agent $i$ with location $x_i$, there exist a profile of other agent locations $\boldsymbol{x}_{-i}$, and reports $\hat{\boldsymbol{x}}_{-i}$ such that for some $x_i'\neq x_i$
\begin{align}\label{equation: uncompromising truthful}
%s':=M(\hat{x}_i, \hat{\boldsymbol{x}}_{-i})\spref_i M(x_i, \hat{\boldsymbol{x}}_{-i}):=s.
u_i^*(M(x_i', \hat{\boldsymbol{x}}_{-i}), \boldsymbol{x}, k)> u_i^*(M(x_i, \hat{\boldsymbol{x}}_{-i}), \boldsymbol{x}, k).
\end{align}
Define $s'=M(x_i', \hat{\boldsymbol{x}}_{-i})$ and $s=M(x_i, \hat{\boldsymbol{x}}_{-i})$. It is immediate from (\ref{equation: uncompromising truthful}) that $s\neq x_i$ and $s\neq s'$. Without loss of generality we assume that $x_i>s$. By assumption, $M$ is a GMM and hence by Lemma~\ref{theorem: generalized median} satisfies the uncompromising property. It follows that $x_i'<s$, since otherwise $x_i' \ge s $ and (\ref{eqaution: uncompromising 1}) would imply $s'=s$ contradicting (\ref{equation: uncompromising truthful}).

Case 1: Suppose $s<s'$. Then $x_i'<s'$ and the uncompromising property (\ref{eqaution: uncompromising 2}) implies that 
$$M(x_i'', \boldsymbol{x}_{-i})=s' \qquad \text{for all } x_i'' \le  s'.$$
If $x_i''\in [s, s']$ the uncompromising property implies that $M(x_i'', \boldsymbol{x}_{-i})=M(x_i, \boldsymbol{x}_{-i})$, i.e., $s'=s$, which contradicts (\ref{equation: uncompromising truthful}). Thus, we conclude that $x_i''<s$.

Now consider a new instance where agent $i$ has true location $y_i=\varepsilon\in (0, s)$, all other agents have true location $y_j=0$ but collectively report $\hat{\boldsymbol{x}}_{-i}$. If agent $i$ reports $y_i=\varepsilon$ then the facility location is $s'$ and $i$ attains utility $1-d(s', \varepsilon)$. If instead agent $i$ reports $y_i'=x_i$ then the facility location is $s<s'$ and $i$ attains strictly higher utility $1-d(s, \varepsilon)$. Thus, the mechanism is not $k$-DIC -- a contradiction. 

Case 2: Suppose $s>s'$. Since $x_i>s>s'$, it follows from the single-peaked property (Proposition~\ref{Proposition: Basic prop 2}) that $u_i^*(s, \boldsymbol{x}, k)\ge  u_i^*(s', \boldsymbol{x}, k)$. This contradicts (\ref{equation: uncompromising truthful}). 
\end{proof}

We now prove our second proposition towards the characterization result. Proposition~\ref{lemma: SP translation} says that, if a mechanism is \stp for some $k<n$ then it is \stp for $k=n$. Thus, the \stp requirement  is more restrictive for $k<n$ than for $k=n$ -- meaning that the capacity constraints induce new strategic concerns for the mechanism designer.

\begin{proposition}\label{lemma: SP translation}
If a mechanism $M$ is \stp, for some $k<n$, then it is \stp for $k=n$. The converse is not true.
\end{proposition}

\begin{proof}
We prove the contrapositive. Suppose that $M$ is not \stp for $k=n$. That is, for some agent $i$ with location $x_i$ there exists a report $x_i'$, a profile of other agent reports $\hat{\boldsymbol{x}}_{-i}$, and a profile of other agent locations $\boldsymbol{x}_{-i}$ such that 
\begin{align}\label{equation: sp translation}
%s':=M(x_i', \hat{\boldsymbol{x}}_{-i})\spref_i M(x_i, \hat{\boldsymbol{x}}_{-i}):=s,
u_i^*(M(x_i', \hat{\boldsymbol{x}}_{-i}), \boldsymbol{x}, n)>u_i^*( M(x_i, \hat{\boldsymbol{x}}_{-i}), \boldsymbol{x}, n).
\end{align}
Let $s'=M(x_i', \hat{\boldsymbol{x}}_{-i})$ and $s=M(x_i, \hat{\boldsymbol{x}}_{-i})$. When $k=n$ all agents are served and so (\ref{equation: sp translation}) simplifies to 
\begin{align}\label{Equation: k<n implies n}
1-d(s', x_i)>1-d(s, x_i).
\end{align}

Now we consider the same profile of reports but for an arbitrary $k<n$. Furthermore, suppose all agents have location equal to $x_i$ and agent $i$ has highest priority ($\triangleright$), i.e., after tie-breaking. The mechanism output is independent of agent true locations and so we still attain $M(x_i', \hat{\boldsymbol{x}}_{-i})=s'$ and $M(x_i', \hat{\boldsymbol{x}}_{-i})=s$. Furthermore, since $i$ has highest priority (recall that the priority is distance-based but in this instance  all agents are equidistant for every facility location) they are always served for every facility location.  In particular, the utility from reporting truthfully is $1-d(s, x_i)$ and misreporting is $1-d(s', x_i)$ -- the latter provides strictly higher utility, as per (\ref{Equation: k<n implies n}). We conclude that the mechanism is not \stp, and since $k<n$ was chosen arbitrarily it holds for all $k<n$. 

The final statement in the proposition was shown in Example~\ref{Example: DIC hard}.
\end{proof}

We now prove our third and final proposition, which completes the characterization result. Proposition~\ref{theorem: capacitated SP equivalence} says that, if a mechanism is \stp for some $k<n$ then it is a GMM.

\begin{proposition}\label{theorem: capacitated SP equivalence}
If a mechanism $M$ is \stp, for some $k<n$, then it is a GMM. %{\color{red} Delete: Thus, in this setting \stp and the uncompromising property are equivalent.  }
\end{proposition}

\begin{proof}
%One direction has been proven in Proposition~\ref{proposition: uncompromise connections}. 
%We prove that when $k<n$, every \stp mechanism is uncompromising which suffices in proving that the mechanism is a GMM. 
Let $M$ be a mechanism that is \stp for some $k<n$. 

First, consider an instance where an arbitrary agent $i$ has location $x_i$, and the other agents report $\hat{\boldsymbol{x}}_{-i}$. If $i$ reports truthfully the mechanism outputs some location that we denote as $s$, i.e., 
\begin{align}\label{Equation: truthful report DIC capacatited leads to GMM s}
s:=M(x_i, \hat{\boldsymbol{x}}_{-i}).
\end{align}
If $s=x_i$ then consider an alternate location and profile of other agents' reports so that the equality does not hold. If no such location and report profile exists then the mechanism always coincides with agent $i$'s report; that is, the mechanism is the agent $i$ dictatorship mechanism, which is a GMM.

%for some $s\neq x_i$ (if no such profile exists then agent $i$'s report dictates the location $s$ and hence the mechanism is trivially uncompromising; \Hau{This is true is because in the definition, it excludes agent with $x_i = s$?}). Wlog assume that $s<x_i$.

Now suppose $s\neq x_i$, and without loss of generality assume $s<x_i$. By assumption $M$ is \stp, for some $k<n$, and so it must be that for all $x_i'$
\begin{align}\label{equation: sp translationx1}
%s':=M(x_i', \hat{\boldsymbol{x}}_{-i})\spref_i M(x_i, \hat{\boldsymbol{x}}_{-i}):=s,
u_i^*( s, \boldsymbol{x}, k)\ge u_i^*(M(x_i', \hat{\boldsymbol{x}}_{-i}), \boldsymbol{x}, k),
\end{align}
where $\boldsymbol{x}$ denotes the location profile of all agents. 

We now show that deviations by agent $i$ satisfy the uncompromising property, i.e., for any $x_i'\ge s$ $M(x_i', \hat{\boldsymbol{x}}_{-i})=s$. To do so, we analyze different cases and sequential refine the possible values of $M(x_i', \hat{\boldsymbol{x}}_{-i})$, we the derive a contradiction to eventually conclude that $M(x_i', \hat{\boldsymbol{x}}_{-i})=s$.\\

\underline{Case 1:} Suppose all other agents have location $s$. When agent $i$ truthfully reports $x_i$ the facility location is $s$ and they attain zero utility. Now consider some report $x_i' \ge s$, leading to facility location 
$$s_{x_i'}:=M(x_i', \hat{\boldsymbol{x}}_{-i}).$$
If $s_{x_i'}\in (\frac{s+x_i}{2}, 1]$ for any $x_i'\ge s$ we attain a contradiction, since this agent $i$ would be served from this report and attain strictly more utility than being truthful. We conclude that 
$$s_{x_i'}\in [0,s)\cup \{s\}\cup (s, \frac{s+x_i}{2}) \qquad \text{ for all } x_i'\ge s.$$

\underline{Case 2:} Suppose all other agents have location $1$, noting that $s<x_i\le 1$. In the event that $x_i=1$ (in which case all agents are equidistant from every facility location), assume agent $i$ has the highest priority in the tie-breaking rule ($\triangleright$). When agent $i$ truthfully reports their location they are served and attain utility $1-d(s, x_i)$. To avoid a contradiction of (\ref{equation: sp translationx1}), it must be that $s_{x_i'}\le s$. Thus, we conclude
\begin{align*}
s_{x_i'}&\in [0,s)\cup \{s\}&& \text{ for all } x_i'\ge s.
\end{align*}

For the sake of a contradiction suppose there exists some $x_i''\ge s$ such that
\begin{align}\label{equation: sp translationx3}
s_{x_i''}&\in [0,s).
\end{align}
Consider a new instance where agent $i$'s location is $y_i=x_i''$ (note that $x_i''\ge s$), all other agents have location $1$, and the other agents report $\hat{\boldsymbol{x}}_{-i}$ (the same profile of reports as per (\ref{equation: sp translationx1})). In the event that $y_i=x_i''=1$ (in which case all agents are equidistant from every facility location), assume agent $i$ has the highest priority in the tie-breaking rule ($\triangleright$). If agent $i$ reports their location $y_i$ the facility location is $s_{y_i}=s_{x_i''}<s$, as per (\ref{equation: sp translationx3}), and they attain utility $1-d(s_{y_i}, y_i)$. But now misreporting to $y_i'=x_i$ then as per (\ref{Equation: truthful report DIC capacatited leads to GMM s}) the facility location is $s$ where 
$$s_{y_i}< s \le y_i,$$
leading to utility $1-d(s, y_i)$. This is a contradiction of the mechanism being \stp, since $d(s, y_i)<d(s_{y_i}, y_i)$; that is, agent $i$ by reporting $y_i'$ instead of their true location $y_i$  attains strictly higher utility. We conclude that $s_{x_i'}=s$ for all $x_i'\ge s$. Thus, the mechanism is uncompromising and hence a GMM.
\end{proof}

%%%%%%%%%%%%%%%%%%%%%%%%%%%%%%%%%%%%%%%%%%%%%%%%

\section{Approximation of \stp mechanisms}\label{Section: approxim of dic}

%%%%%%%%%%%%%%%%%%%%%%%%%%%%%%%%%%%%%%%%%%%%%%%%

Given the characterization result (Theorem~\ref{Corollary: equivalence results}) of the previous section, there is no distinction between the family of mechanisms that are \stp for some $k<n$, and the family of mechanisms that are \stp for all $k\le n$: both families are equal to the GMM family. Accordingly, we will now simply refer to a mechanism as being \stp.

%abusing notation slightly we will refer to a mechanism which $k$-\stp all $k\le n$ as simply a \stp mechanism. Note that this is equivalent to simply being $k$-\stp for some $k<n$, and the mechanism being a GMM.

\subsection{Optimal mechanism is not \stp}

We first show that, in general for $k<n$, the optimal mechanism is not \stp. Note that this result contrasts with the $k=n$ setting where the median mechanism is both optimal and \stp (Remark~\ref{Remark: special case}).

\begin{theorem}\label{Theorem: optimal not dic}
The optimal mechanism is \stp if and only if $k\in \{1,n\}$.
%For $k\notin\{1,n\}$, the welfare-optimal mechanism is not \stp. For $k\in \{1,n\}$ the welfare-optimal mechanism is \stp.
\end{theorem}

\begin{proof}
The backward direction of the theorem statement is straightforward: If $k=1$ then for any $i\in N$ the agent $i$ dictator mechanism, where the mechanism output always coincides with agent $i$'s report, is both optimal and \stp. This is trivial and we do not provide further details.  If $k=n$ then the median mechanism is both optimal and \stp. This result has long been known and can be found in~\citet{Blac48,Moul80,PrTe13}.

We now prove the forward direction using the contrapositive. Let $k\notin \{1,n\}$ and partition the agent into $\floor{n/k}$ groups of size $k$, denoted by $N_{t}$ for $t=1, 2, \ldots, \floor{n/k}$, and one group of size $n-\floor{n/k}$, denoted by $N_{\floor{n/k}+1}$.  We now identify $\floor{n/k}+1$ locations in $[0,1]$, let 
$$y_t=  \frac{t}{\floor{n/k}+1}\qquad \text{ for } t=1, 2, \ldots, \floor{n/k}+1.$$

Consider a scenario such that for each $t=1, 2, \ldots, \floor{n/k}+1$, all but one agent in $N_t$ is located at $y_t$ and a single agent is located at $y_t- t\, \varepsilon$ for some sufficiently small $\varepsilon>0$. In each instance denote the single agent located at $y_t- t\, \varepsilon$ by $i_t\in N_t$.

In this scenario it is immediate the optimal welfare is attained by locating the facility at location $y_1$, leading to a social welfare of $k-\varepsilon$ and agent $i_1$ attain utility $1-\varepsilon$.

Now in a new scenario where agent $i_1$ is located at $y_1-3\varepsilon$ the optimal mechanism must locate the facility at $y_2$. In this case agent $i_1$ attains utility zero. However, if agent $i_1$ misreport their location to $y_1-\varepsilon$ then (as shown above) the facility location will be $y_1$ and they will attain strictly higher utility $1-\varepsilon$. That is, the optimal mechanism is not \stp for $k\notin \{1,n\}$.
\end{proof}

Despite Theorem~\ref{Theorem: optimal not dic} stating a stark impossibility result, we note that absent strategic manipulations by the agents the optimal mechanism can  be efficiently computed. Remark~\ref{Remark: Computational result} says that,  for any $k\le n$ the optimal mechanism's output and corresponding welfare can be computed in polynomial time.

\begin{remark}\label{Remark: Computational result}
The optimal facility location and welfare can be computed in polynomial time for any $k\le n$.
\end{remark}

We sketch an informal argument for Remark~\ref{Remark: Computational result}. Order the agents $i\in  N$ such that $x_i\le x_{j}$ if and only if $i\le j$. It is straightforward to show that an optimal solution has two features (1) the facility serves a  \emph{contiguous} set of $k$ agents, i.e., if agent $i$ and $i+2$ are served then agent $i+1$ is served, and (2) the facility is located at the median of these $k$ served agents. Given these features, it is immediate that a polynomial-time procedure exists by simply comparing the welfare produced by, the at most $n$, sets of $k$ contiguous agents.  

%\edit{Do we have a reference we can give? Can we provide a sketch of the argument?}

\subsection{Lower bound on \stp approximation}

Utilizing the characterization result of \stp mechanisms via the family of GMMs, we provide a lower bound on the approximation ratio for all \stp mechanisms. 

Theorem~\ref{theorem: uncompromising at best 2} shows that at best a \stp mechanisms provides a $2\frac{k}{k+1}$-approximation when $k\le \ceil{(n-1)/2}$, and otherwise provides at best an $\max\{\frac{n-1}{k+1},1\}$-approximation. This lower bound on the approximation ratio is illustrated in Figure~\ref{figure: illustration DIC lower boundxxy}.

\begin{theorem}\label{theorem: uncompromising at best 2}
Let $n\ge 2$. A \stp mechanism is at best an $\alpha$-approximation with $\alpha=2\frac{k}{k+1}$ when $1\le k\le \ceil{(n-1)/2}$, and $\alpha=\max\{\frac{n-1}{k+1},1\}$ otherwise. 
\end{theorem}

\begin{proof}
Let $M$ be a \stp mechanism, and consider a scenario where all $n$ agents have distinct locations contained in the interval $I=(1/2-1/2\varepsilon, \ 1/2 +1/2\varepsilon)$ for some sufficiently small $\varepsilon>0$. Denote the profile of agent locations by $\boldsymbol{x}$, and the mechanism's corresponding output by $s=M(\boldsymbol{x})$. 

We consider two cases.

\underline{Case 1:} Suppose $s\notin I$ and without loss of generality assume $s<1/2-1/2\varepsilon$. Now suppose that sequentially agents $i=1, 2, \ldots, n$ have their locations changed and kept at $x_i=1$, and consider the sequence of facility locations produced by the mechanism $s_1, s_2, \ldots, s_n$. By the uncompromising property (satisfied by $M$ since it is a GMM) the location of the facility never changes from $s$. That is, $s_n=s$ despite every agent having location at $1$. The optimal social welfare in this scenario is clearly $k$, however, the mechanism provides welfare of 
\begin{align*}
k(1-d(s, 1))&=k\, s< k (1/2-1/2\varepsilon)\rightarrow k/2 &&\text{as $\varepsilon\rightarrow 0$.}
\end{align*}
Thus, the approximation ratio is at best $k/(k/2)=2$.

\underline{Case 2:} Suppose $s\in I$ and without loss of generality assume $s\le 1/2$. Let $\lambda_1, \lambda_2$ be the number of agents with locations strictly less than $s$, and strictly above $s$, respectively. Note that $\lambda_1+\lambda_2\in \{n-1, n\}$. Similar to Case 1, suppose the $\lambda_1$ agents instead had location at $0$ and the $\lambda_2$ agents had location at $1$ -- by the uncompromising property the facility location is unchanged.  

To attain the bound on the approximation ratio we consider two subcases where $k\le \ceil{(n-1)/2}$ and $k>\ceil{(n-1)/2}$.

 In the first subcase ($k\le \ceil{(n-1)/2}$): the optimal welfare is $k$, since either $\lambda_1$ or $\lambda_2$ exceeds $k$ meaning that $k$ agents can be served at either 0 or 1. The mechanism's welfare is at most
\begin{align*}
1+(k-1) (1-d(s,0))&< 1+(k-1) (1/2-1/2\varepsilon)\rightarrow 1/2 + k/2&&\text{as $\varepsilon\rightarrow 0$.}
\end{align*}
Thus, the approximation ratio is at best $k/(1/2+k/2)=2 \, k/(k+1)$.

In the second subcase ($k> \ceil{(n-1)/2}$): the optimal welfare is at worst $\ceil{(n-1)/2}$, i.e., when the facility serves either $\lambda_1$ or $\lambda_2$ agents  (whichever is larger) from location 0 or 1. The mechanism's welfare is at most
\begin{align*}
1+ (k-1)(1-d(0, s))&<k-(k-1)(1/2-1/2\varepsilon)\rightarrow k/2+1/2&&\text{as $\varepsilon\rightarrow 0$.}
\end{align*}
Thus, the approximation ratio is at best $\ceil{(n-1)/2}/(k/2+1/2)$, but 
\begin{align*}
\ceil{(n-1)/2}/(k/2+1/2)&\ge \frac{(n-1)/2}{(k+1)/2}= \frac{n-1}{k+1}.
\end{align*}
Furthermore since $k> (n-1)/2$ it follows that $ \frac{n-1}{k+1}<2$. Of course, this bound is only meaningful when $n-1/k+1 >1$.

We conclude that when $k\le \ceil{(n-1)/2}$ the approximation ratio is at best $2\frac{k}{k+1}$ and otherwise is at best $\max\{\frac{n-1}{k+1},1\}$.
\end{proof}

\subsection{Optimized approximation ratio for \stp Mechanism}

We now  analyze the performance of the median mechanism for general $k\le n$. In instances where $k\in\{1,n\}$, the median mechanism is both optimal mechanism and \stp (Theorem~\ref{Theorem: optimal not dic}). Furthermore, this mechanism is \stp for all $k\le n$ since the median mechanism is a GMM (Theorem~\ref{Corollary: equivalence results}).

%\Hau{Can we add a sentence or two saying why we just jump into the median mechanism? ... Maybe like: 
%Since the median mechanism is optimal and DIC in the classical setting, we begin by analyzing the median mechanism. 
%While the median mechanism is not optimal for a wide range of k $<$ n, it is DIC for any k. 
%As k varies, the median mechanism obtains different (worst-case) approximation ratio. 
%: A quick note, I would assume that the reviewers would know about 
%DIC = strategyproofness ... not sure if we want to say that somewhere in the beginning, otherwise they might ask some question about it }

%Formally, the median mechanism is defined as follows. 
%\begin{definition}
%The \emph{median mechanism} takes a profile of reports $\hat{\boldsymbol{x}}$ and outputs the location corresponding to the $\floor{(n+1)/2}$th smallest location in $\hat{\boldsymbol{x}}$. We denote the mechanism by $median(\hat{\boldsymbol{x}})$.
%\end{definition}

%\begin{proposition}
%The median mechanism is \stp. %\truth
%\end{proposition}
%
%\begin{proof}
%The median mechanism is GMM.
%\end{proof}

Theorem~\ref{Theorem: median performance} says that the median mechanism is an $\alpha$-approximation where $\alpha=2\frac{k}{k+1}$ when $k\le \floor{(n+1)/2}$, and $\alpha=\min\{2\frac{k}{k+1}, 1+2\frac{n-k+1}{3k-2n-2}\}$ otherwise. In particular, this means that the median mechanism is optimal among \stp mechanism for $k\le \floor{(n-1)/2}$ since the  approximation-ratio matches the lower bound found in Theorem~\ref{theorem: uncompromising at best 2}. These approximation results are illustrated in Figure~\ref{figure: illustration DIC lower boundxxy}.
%\Hau{What happen to n < 5?}

\begin{theorem}\label{Theorem: median performance}
Let $n\ge 5$. The median mechanism is an $\alpha$-approximation with $\alpha=2\frac{k}{k+1}$ for $k\le \floor{(n+1)/2}$, and $\alpha=\min\{2\frac{k}{k+1}, 1+2\frac{n-k+1}{3k-2n-2}\}$ otherwise.%
\end{theorem}

\begin{proof}
Let $n\ge 5$. Throughout the proof let $i_m$ denote the agent with median location (choose the agent arbitrarily if multiple such agents exist), and let $s_m$ denote the median location. The median mechanism provides welfare
\begin{align*}
\Pi_M(\boldsymbol{x}, k)=\max_{N' \in N_{k}} \sum_{i\in N_{k}} (1-d(s_m, x_i))=1+\max_{N' \in N_{k-1, i_m}} \sum_{i\in N_{k-1, i_m}} (1-d(s_m, x_i)),
\end{align*}
where $N_k$ is the set of all $k$-sized subsets of $N$ and $N_{k-1, i_m}$ is the set of all $(k-1)$-sized subsets of $N\backslash \{i_m\}$. This follows since the subset of agents served are always the $k$-closest to the facility location. Hence, given a facility location, the served subset is welfare maximizing. Furthermore, the median location coincides with at least one agent's location, i.e., agent $i_m$.\\

First, we provide an upper bound on the approximation-ratio for all $k$. The median mechanism locates the facility at the $\floor{(n+1)/2}$-th location and hence there are $\floor{(n+1)/2}-1$ agents with locations (weakly) below and $\ceil{(n+1)/2}-1$ with locations (strictly) above. A lower bound on the median mechanism's welfare is attained when the agents below and above the median location at located at 0 and 1, respectively. Thus,  
$$\Pi_M(\boldsymbol{x}, k)\ge 1+ (k-1)\max\{1-d(s_m, 0),\ 1-d(s_m,1)\},$$
and since either $d(s_m,0)\le 1/2$ or $d(s_m, 1)\le 1/2$ it follows that $\Pi_M(\boldsymbol{x}, k)\ge (k+1)/2$. This leads to an upper bound on the approximation-ratio of $k/((k+1)/2)=2\frac{k}{k+1}$ for all $k$, since the optimal welfare is always bounded above by $k$.\\

Now we attain a tighter upper bound for certain values of $k$. To do so, we bound the median welfare using the optimal welfare. Let $s^*$ be the location of the facility under the optimal mechanism. Let $N_m^*$ denote the set of $k$ agents served under the median mechanism,  and let $N^*$ denote the set of $k$ agents served under the optimal mechanism. We have 
\begin{align*}
\Pi_M(\boldsymbol{x}, k)&\ge  \sum_{i\in N^*} \big(1-d(s_m, x_i)\big)\\
& =\sum_{i\in N^*} \Big(1-d(s_m, x_i) -d(s^*, x_i) +d(s^*, x_i) \Big)\\
&=\Pi^*(\boldsymbol{x}, k)-\sum_{i\in N^*} \Big(d(s_m, x_i)-d(s^*, x_i)\Big).
\end{align*}
Clearly, the lower bound is smallest when $s_m\neq s^*$, without loss of generality assume that $s_m<s^*$. Let $N_1^*, N_2^*$ be a partition of $N^*$ such that $|N_1^*|, |N_2^*|\le \floor{(n+1)/2}$ and all agents in $N_1^*$ have location in $[0,s_m]$ and agent in $N_2^*$ have location in $[s_m, 1]$. Such a partition of $N^*$ exists since the location $s_m$ coincides with the $ \floor{(n+1)/2}$ highest location.
%Let $N_1^*, N_2^*$ by a partition of $N^*$ such that agents in $N_1^*$ have location (weakly) less than $s_m$ and agents in $N_2^*$ have location strictly higher than $s_m$. 
Using this partition we further bound the median mechanism's welfare: 
\begin{align*}
\Pi_M(\boldsymbol{x}, k)&\ge  \Pi^*(\boldsymbol{x}, k)-\sum_{i\in N_1^*} (d(s_m, x_i)-d(s^*, x_i))-\sum_{i\in N_2^*} (d(s_m, x_i)-d(s^*, x_i))\\
&\ge  \Pi^*(\boldsymbol{x}, k)- |N_1^*| \max_{x\in [0, s_m]} \Big(s_m- s^*-2x\Big)-|N_2^*| \max_{x\in [s_m,1]}\Big(x_i-s_m-|s^*-x_i|\Big)\\
&\ge  \Pi^*(\boldsymbol{x}, k)-|N_1^*| (s_m-s^*)-|N_2^*| (s^*-s_m)\\
&\ge  \Pi^*(\boldsymbol{x}, k)-(|N_2^*|-|N_1^*|) (s^*-s_m)\\
&\ge  \Pi^*(\boldsymbol{x}, k)-(|N_2^*|-|N_1^*|).
\end{align*}
We now attain our lower bound by considering the maximum value of $|N_2^*|-|N_1^*|$. For $k\le \floor{(n+1)/2}$, the value can only be guaranteed to be no larger than $k$ -- leading to a trivial zero lower on $\Pi_M(\boldsymbol{x}, k)$. However, for $k>\floor{(n+1)/2}$ we attain a more useful bound by noting that
$$(|N_2^*|-|N_1^*|) \le \floor{(n+1)/2}-(k-\floor{(n+1)/2})=2\floor{(n+1)/2}-k\le n+1-k.$$
This lead to an approximation-ratio upper bound of
$$\max_{\boldsymbol{x}\in \prod_{i=1}^n X}\Bigg\{\frac{  \Pi^*(\boldsymbol{x}, k)}{  \Pi^*(\boldsymbol{x}, k)-n-1+k}\Bigg\}=\max_{\boldsymbol{x}\in \prod_{i=1}^n X}\Bigg\{1+\frac{ n+1-k}{  \Pi^*(\boldsymbol{x}, k)-n-1+k}\Bigg\}.$$
Furthermore, for any instance $\Pi^*(\boldsymbol{x}, k)\ge k/2$ since at least as much welfare is attained by locating the facility at $s=1/2$. Thus, an upper bound on the approximation ratio is 
\begin{align*}1+\frac{ n+1-k}{  k/2-n-1+k}=1+2\frac{ n+1-k}{  3k-2n-2}.
\end{align*} \end{proof}

\section{Extension: Location-Allocation Mechanisms }\label{Section: Excludable}

%%%%%%%%%%%%%%%%%%%%%%%%%%%%%%%%%%%%%%%%%%%%%%%%

In this section we consider an extension of our framework where the mechanism designer is able to dictate which agents are served by the facility. Note that this extension introduces an underlying assumption that the facility is excludable. In practice, a designer may be able to dictate which agents are served by issuing permits or, when costs are not  prohibitive, checking the identities of agents attempting to benefit from the facility. 

% where the good is excludable. That is, we consider a setting where the mechanism designer has the ability to issue \emph{permits} which restricts access to the facility to only a subset of at most $k$ agents. 

Previously, a mechanism $M \ : \prod_{i\in N} X\rightarrow X$ was defined as a function mapping a profile of locations to a single facility location. In our extension, a mechanism not only locates the facility but also chooses a subset of at most $k$ agents to be served by the facility, if they so choose. We denote these extended mechanisms by 
$$M_A \ : \prod_{i\in N} X\rightarrow X \times N_k,$$
where $N_k=\{A\subseteq N \ : 0<|A|\le k\}$. We call these mechanisms \emph{location-allocation} mechanisms, to distinguish them from the (location-only) mechanisms considered in earlier sections of the present paper. The output of the mechanism is a pair $(s,A)\in X\times N_k$ where $s\in X$ denotes the facility location and $A\in N_k$ denotes the subset of agents allocated to the facility. Abusing notation slightly we will denote the mechanism output from a location profile $\boldsymbol{x}$ by $s_{\boldsymbol{x}}$ and $A_{\boldsymbol{x}}$ where $M_A (\boldsymbol{x})=\big(s_{\boldsymbol{x}},\ A_{\boldsymbol{x}}\big)$. An agent $i\in A$ is guaranteed to be served by the facility if they so choose, whilst an agent $i\notin A$ is never served. 

We omit the details, but it is immediate that the modified subgame $\Gamma_{\boldsymbol{x}}(s, k, A)$ has an essentially unique ex-post Nash equilibrium where all agents $i\in A$ are served by the facility and the remaining agents are not. Thus, we assume that agent $i$ reports their location to the mechanism designer with the understanding that they will be served by the facility if and only if $i\in A$, as per the ex-post Nash equilibrium. Note that the strategyproof concept, \stp, in this section still coincides with the concept used in the earlier sections, albeit with the modified subgame explained above.

We first remark that the revelation principle~\cite{Gibb73} does not apply. A location-only mechanism, based on the profile of agent reports, $\hat{\boldsymbol{x}}$, outputs a facility location $s$ -- that depends on $\hat{\boldsymbol{x}}$ -- and a subset of $k$ agents are then allocated to the facility, via the ex-post Nash equilibrium, $A\subseteq N$ -- this subset depends on the agent true locations $\boldsymbol{x}$ and not the reports $\hat{\boldsymbol{x}}$. In contrast, a location-allocation mechanism outputs both a facility location and an allocation of $\le k$ agents to the facility depending on agent reports $\hat{\boldsymbol{x}}$, and not true locations $\boldsymbol{x}$. Thus, an agent misreporting their location -- in a way that does not affect the facility location -- will never affect whether or not they are served by the facility under a location-only mechanisms. However, under a location-allocation mechanism the agent may potentially benefit from the misreport if they are now allocated to the facility by the mechanism.

We now show that no `reasonable' location-allocation-mechanism is \stp. In particular, we only enforce one criteria, which is a weak form of \emph{anonymity}. Informally speaking, we require that the location-allocation mechanism allocates agents to the facility independently of their label if their report is distinct from all other agents. The usual definition of anonymity is not directly applicable since with a deterministic mechanism, if all agents report identical locations the mechanism must discriminate against at least $n-k$ agents who will not be included in the allocation set $A$.

To formally define our anonymity condition we first introduce the notion of an $i$-identifiable location profile. This is simply a profile where agent $i$ is uniquely identified by their report.

\begin{definition}[$i$-identifiable location profile]Let $i\in N$. A location profile $\boldsymbol{x}$ is \emph{$i$-identifiable} if 
$$x_i\neq x_j \qquad \text{  for  all } j\in N\backslash \{i\}.$$
\end{definition}

We now define our anonymity condition, which we call {\em allocation-anonymous} since the condition only applies to the allocation set rather than the facility location. Informally speaking, the allocation-anonymous condition requires that for every $i$-identifiable location profile, whether or not agent $i$ is allocated to the facility does not depend on $i$'s label. Given that allocation-anonymity only applies to $i$-identifiable location profiles the condition is relatively weak.

\begin{definition}[Allocation-anonymous]\label{allocation anon definition} The mechanism $M_A$ is said to be \emph{allocation-anonymous} if for  every distinct $i,j\in N$ and every $i$-identifiable location profile $\boldsymbol{x}$, the modified profile $\boldsymbol{x}'$ such that $x_\ell=x_\ell' $ for all $\ell\neq i,j$ and
$$x_i'=x_j \qquad \text{ and }\qquad x_j'=x_i,$$
we have 
$$i\in A_{\boldsymbol{x}} \iff j\in A_{\boldsymbol{x}'},$$
where $A_{\boldsymbol{x}}$ is such that $M_A (\boldsymbol{x})=\big(s_{\boldsymbol{x}},\ A_{\boldsymbol{x}}\big)$.
\end{definition}

%\begin{remark}
%It is straight forward to see that any dictatorial rule which 
%
%are not allocation-anonymous. 
%\end{remark}

We now show that if we restrict our attention to allocation-anonymous mechanisms there is no \stp location-allocation mechanism.

\begin{theorem}\label{theorem: impossibility of location-allocation}
Let $k<n$, any location-allocation mechanism $M_A$ that is allocation-anonymous is not \stp.
\end{theorem}

\begin{proof}
For the sake of a contradiction suppose that $M_A$ is a location-allocation mechanism that is both allocation-anonymous and \stp.

First consider a location profile $\boldsymbol{x}$ where $x_i=3/4$ for all $i\in N$, and denote the output of the mechanism by $M_A(\boldsymbol{x})=(s,A)$. Let $i^*$ be some agent such that $i^*\in A$ and $j^*$ some agent such that $j^*\notin A$. In this outcome agent $j^*$ attains utility zero, since $j^*\notin A$.

Now consider another location profile $\boldsymbol{x}'$ such that $x_i'=3/4$ for all $i\in N\backslash \{j^*\}$ and  $x_{j^*}'=1/2$. Note that the profile $\boldsymbol{x}'$ can be achieved via a unilateral deviation from the profile $\boldsymbol{x}$ by agent $j^*$. Denote the mechanism's output from this location profile by $M_A(\boldsymbol{x}')=(s', A')$.  We consider two cases and derive a contradiction in each case.\\

Case 1: Suppose $j^*\in A'$ and suppose that $\boldsymbol{x}$ is the true location of all agents. In this case agent $j^*$ by misreporting their location to $x_{j^*}'=1/2$ strictly profits, since under the profile $\boldsymbol{x'}$ we have $j^*\in A'$ and their utility is now $1-d(s',\frac{3}{4})>0$ rather than zero. Thus, we have a contradiction. \\

Case 2: Suppose $j^*\notin A'$ and suppose that the agents have true locations $y_\ell=3/4$ for all $\ell \in N\backslash \{i^*\}$ and $y_{i^*}=1/2$. Denote the mechanism's output from location profile $\boldsymbol{y}$ by $M_A(\boldsymbol{x}')=(s'', A'')$. Notice that $\boldsymbol{x}'$ is a $j^*$-identifiable location profile and the profile $\boldsymbol{y}$ satisfies the condition in Definition~\ref{allocation anon definition}, and so by the allocation-anonymous property we require that 
$$i^*\in A'' \iff j^*\in A'.$$
Thus, we infer that $i^*\notin A''$ and attain zero utility under the location profile $\boldsymbol{y}$. Now suppose agent $i^*$ unilaterally deviates and reports the location $y_{i^*}'=3/4$. In this case, the location profile coincides with the profile $\boldsymbol{x}$ where $x_\ell=3/4$ for all $\ell\in N$. But recall that $M_A(\boldsymbol{x})=(s,A)$ and $i^*$ was taken to be some agent such that $i^*\in A$. Thus, under this unilaterally misreport agent $i^*$ is now served and attain strictly positive utility of $1-d(s, y_{i^*})>0$. This is a profitable deviation and contradicts our assumption that the mechanism $M_A$ was \stp. We conclude that there is no location-allocation mechanism that is both allocation-anonymous and \stp.
\end{proof}

The above impossibility result means that the extensive-form approach taken in the main body of this paper is crucial for \stp mechanisms that are non-dictatorial. The use of an extensive-form game and corresponding ex-post Nash equilibria to decide the allocation of agents to the facility reduces the incentive compatibility constraints faced by the mechanism designer. Furthermore, the result suggests that the excludability of the facility presents a greater challenge for incentive compatibility than rivalry.

%%%%%%%%%%%%%%%%%%%%%%%%%%%%%%%%%%%%%%%%%%%%%%%%

\section{Discussion and Conclusion}\label{section: discussion and conlc}

%%%%%%%%%%%%%%%%%%%%%%%%%%%%%%%%%%%%%%%%%%%%%%%%

We now conclude the paper with a brief discussion of future research directions.

%\underline{Assignment and allocation mechanisms:}  Our focus was restricted to mechanisms which provided a location for the facility but did not dictate which agents are served by the facility. Instead, the subset of served agents were derived from the equilibrium assignment of an induced subgame. A larger family of mechanisms can be considered which provide both a location and a subset of agents to be served by the facility. Under reasonable anonymity type assumption on the mechanism it can be shown that no \stp mechanism exists which both locates the facility and dictates which agents are served. Exploring ways to avoid this impossibility result may be an interesting area of future work.\\

\underline{Extensions to multiple facilities:} In the present paper we focused on the case of a single capacity constrained facility location problem. Extending the capacity constrained to multiple facilities presents a number of challenges. Firstly, the subgame induced from a profile of facility location will lead to multiple equilibria that are not welfare (nor utility) equivalent. Furthermore, even when ignoring the multiplicity of equilibria issues, the mechanism design problem is drastically more complicated -- as is the algorithmic problem of finding the optimal facility locations (see Brimberg et al.~\cite{BKE-C+01}). A recent contribution by Golowich, Narasimhan and Parkes~\cite{GNP18} explores the mechanism design problem for multiple facilities without capacity constraints.

\underline{Weakening \stp:} A natural direction to consider is weakening the strategproofness concept (\stp) that we use in the present paper. The \stp requirement is very strong: agents must attain maximal ex-post utility from reporting their location no matter what other agents report, and other agents' true locations.  The weaker notion of \emph{ex-post Incentive Compatible (IC)} may be interesting to be explore for both characterization and performance results. This notion requires that agents attain maximal ex-post utility from reporting their location no matter the other agents' true locations, but conditional on the other agents reporting truthfully. It is straightforward to construct IC mechanisms that out-perform the median mechanism for certain parameter ranges.

\textbf{Conclusion:} In this paper we initiated the study of the capacity constrained facility location problem from a mechanism design perspective. We formalized a model that allows the subset of served agents to be endogenously derived from equilibrium outcomes. Our main contribution is a complete characterization of all \stp mechanisms via the family of GMM mechanisms. This characterization also provides a novel perspective to an open problem in regard to GMM mechanisms, posed in~\cite{BoJo83}. Our second contribution is an analysis of the performance of \stp mechanisms with respect to social welfare -- where we also show that the well-known median mechanism is optimal among \stp mechanism for certain parameter ranges. Finally, we show that extending the space of mechanisms to allow the mechanism to allocate agents to the facility leads to a  stark impossibility result. Namely, there is no allocation-anonymous \stp mechanism which both locates the facility and stipulates the subset of agents to be served.

% Bibliography
\bibliographystyle{ACM-Reference-Format}
 % \bibliography{../../pamas/abb,../../pamas/group,../../pamas/brandt,../../pamas/aziz}
%

%  \bibliography{adtbib_bl/abb_bl.bib,adtbib_bl/adt_bl.bib}
  \bibliography{abb_bl.bib,adt_bl.bib}


\begin{thebibliography}{00}

%%% the Web address.

\ifx \showCODEN    \undefined \def \showCODEN     #1{\unskip}     \fi
\ifx \showDOI      \undefined \def \showDOI       #1{{\tt DOI:}\penalty0{#1}\ }
  \fi
\ifx \showISBNx    \undefined \def \showISBNx     #1{\unskip}     \fi
\ifx \showISBNxiii \undefined \def \showISBNxiii  #1{\unskip}     \fi
\ifx \showISSN     \undefined \def \showISSN      #1{\unskip}     \fi
\ifx \showLCCN     \undefined \def \showLCCN      #1{\unskip}     \fi
\ifx \shownote     \undefined \def \shownote      #1{#1}          \fi
\ifx \showarticletitle \undefined \def \showarticletitle #1{#1}   \fi
\ifx \showURL      \undefined \def \showURL       #1{#1}          \fi
% invisible to TeX
\providecommand\bibfield[2]{#2}
\providecommand\bibinfo[2]{#2}
\providecommand\natexlab[1]{#1}
\providecommand\showeprint[2][]{arXiv:#2}

\bibitem[\protect\citeauthoryear{Abdulkadiro\u{g}lu and
  S\"{o}nmez}{Abdulkadiro\u{g}lu and S\"{o}nmez}{2003}]%
        {AbSo03}
\bibfield{author}{\bibinfo{person}{A. Abdulkadiro\u{g}lu} {and}
  \bibinfo{person}{T. S\"{o}nmez}.} \bibinfo{year}{2003}\natexlab{}.
\newblock \showarticletitle{School Choice: {A} mechanism Design Approach}.
\newblock \bibinfo{journal}{{\em American Economic Review\/}}
  \bibinfo{volume}{93}, \bibinfo{number}{3} (\bibinfo{year}{2003}),
  \bibinfo{pages}{729--747}.
\newblock


\bibitem[\protect\citeauthoryear{Ashlagi and Roth}{Ashlagi and Roth}{2011}]%
        {AsRo11}
\bibfield{author}{\bibinfo{person}{I. Ashlagi} {and} \bibinfo{person}{A.
  Roth}.} \bibinfo{year}{2011}\natexlab{}.
\newblock \showarticletitle{Individual rationality and participation in large
  scale, multi-hospital kidney exchange}. In \bibinfo{booktitle}{{\em
  Proceedings of the 12th ACM Conference on Electronic Commerce (ACM-EC)}}.
  \bibinfo{publisher}{ACM Press}, \bibinfo{pages}{321--322}.
\newblock


\bibitem[\protect\citeauthoryear{Barbar{\`a}, Mass{\'o}, and
  Serizawa}{Barbar{\`a} et~al\mbox{.}}{1998}]%
        {BMS98}
\bibfield{author}{\bibinfo{person}{S. Barbar{\`a}}, \bibinfo{person}{J.
  Mass{\'o}}, {and} \bibinfo{person}{S. Serizawa}.}
  \bibinfo{year}{1998}\natexlab{}.
\newblock \showarticletitle{Strategy-proof voting on compact ranges}.
\newblock \bibinfo{journal}{{\em Games and Economic Behavior\/}}
  \bibinfo{volume}{25} (\bibinfo{year}{1998}), \bibinfo{pages}{272--291}.
\newblock


\bibitem[\protect\citeauthoryear{Black}{Black}{1948}]%
        {Blac48}
\bibfield{author}{\bibinfo{person}{D. Black}.} \bibinfo{year}{1948}\natexlab{}.
\newblock \showarticletitle{On the rationale of group decision-making}.
\newblock \bibinfo{journal}{{\em Journal of Political Economy\/}}
  \bibinfo{volume}{56}, \bibinfo{number}{1} (\bibinfo{year}{1948}),
  \bibinfo{pages}{23--34}.
\newblock


\bibitem[\protect\citeauthoryear{Border and Jordan}{Border and Jordan}{1983}]%
        {BoJo83}
\bibfield{author}{\bibinfo{person}{K.~C. Border} {and} \bibinfo{person}{J.~S.
  Jordan}.} \bibinfo{year}{1983}\natexlab{}.
\newblock \showarticletitle{Straightforward Elections, Unanimity and Phantom
  Voters}.
\newblock \bibinfo{journal}{{\em The Review of Economic Studies\/}}
  \bibinfo{volume}{50}, \bibinfo{number}{1} (\bibinfo{year}{1983}),
  \bibinfo{pages}{153--170}.
\newblock


\bibitem[\protect\citeauthoryear{Brandeau and Chiu}{Brandeau and Chiu}{1989}]%
        {BrCh89}
\bibfield{author}{\bibinfo{person}{M.~L. Brandeau} {and} \bibinfo{person}{S.~S.
  Chiu}.} \bibinfo{year}{1989}\natexlab{}.
\newblock \showarticletitle{An overview of representative problems in location
  research}.
\newblock \bibinfo{journal}{{\em Management Science\/}} \bibinfo{volume}{35},
  \bibinfo{number}{6} (\bibinfo{year}{1989}), \bibinfo{pages}{645--674}.
\newblock


\bibitem[\protect\citeauthoryear{Brimberg, Korach, Eben-Chaim, and
  Mehrez}{Brimberg et~al\mbox{.}}{2001}]%
        {BKE-C+01}
\bibfield{author}{\bibinfo{person}{J. Brimberg}, \bibinfo{person}{E. Korach},
  \bibinfo{person}{M. Eben-Chaim}, {and} \bibinfo{person}{A. Mehrez}.}
  \bibinfo{year}{2001}\natexlab{}.
\newblock \showarticletitle{The capacitated $p$-facility location problem on
  the real line}.
\newblock \bibinfo{journal}{{\em International Transactions in Operational
  Research\/}}  \bibinfo{volume}{8} (\bibinfo{year}{2001}),
  \bibinfo{pages}{727--738}.
\newblock


\bibitem[\protect\citeauthoryear{Charikar, Guha, Tardos, and Shmoys}{Charikar
  et~al\mbox{.}}{2002}]%
        {Charikar02:Constant}
\bibfield{author}{\bibinfo{person}{M. Charikar}, \bibinfo{person}{S. Guha},
  \bibinfo{person}{E. Tardos}, {and} \bibinfo{person}{D.B. Shmoys}.}
  \bibinfo{year}{2002}\natexlab{}.
\newblock \showarticletitle{A constant-factor approximation algorithm for the
  k-median problem}.
\newblock \bibinfo{journal}{{\it J. Comput. System Sci.}} \bibinfo{volume}{65},
  \bibinfo{number}{1} (\bibinfo{year}{2002}), \bibinfo{pages}{129--149}.
\newblock


\bibitem[\protect\citeauthoryear{Ching}{Ching}{1997}]%
        {Chin97}
\bibfield{author}{\bibinfo{person}{S. Ching}.} \bibinfo{year}{1997}\natexlab{}.
\newblock \showarticletitle{Strategy-proofness and ``{M}edian {V}oters"}.
\newblock \bibinfo{journal}{{\em International Journal of Game Theory\/}}
  \bibinfo{volume}{26} (\bibinfo{year}{1997}), \bibinfo{pages}{473--490}.
\newblock


\bibitem[\protect\citeauthoryear{Cygan, Hajiaghayi, and Khuller}{Cygan
  et~al\mbox{.}}{2012}]%
        {Cygan12:LP}
\bibfield{author}{\bibinfo{person}{M. Cygan}, \bibinfo{person}{M.~T.
  Hajiaghayi}, {and} \bibinfo{person}{S. Khuller}.}
  \bibinfo{year}{2012}\natexlab{}.
\newblock \showarticletitle{LP rounding for k-centers with non-uniform hard
  capacities}. In \bibinfo{booktitle}{{\em Foundations of Computer Science
  (FOCS), 2012 IEEE 53rd Annual Symposium on}}. IEEE,
  \bibinfo{pages}{273--282}.
\newblock


\bibitem[\protect\citeauthoryear{Feldman, Fiat, and Golomb}{Feldman
  et~al\mbox{.}}{2016}]%
        {FFG16}
\bibfield{author}{\bibinfo{person}{M. Feldman}, \bibinfo{person}{A. Fiat},
  {and} \bibinfo{person}{I. Golomb}.} \bibinfo{year}{2016}\natexlab{}.
\newblock \showarticletitle{On voting and facility location}. In
  \bibinfo{booktitle}{{\em Proceedings of the 17th ACM Conference on Electronic
  Commerce (ACM-EC)}}. \bibinfo{publisher}{ACM Press},
  \bibinfo{pages}{269--286}.
\newblock


\bibitem[\protect\citeauthoryear{Gibbard}{Gibbard}{1973}]%
        {Gibb73}
\bibfield{author}{\bibinfo{person}{A. Gibbard}.}
  \bibinfo{year}{1973}\natexlab{}.
\newblock \showarticletitle{Manipulation of voting schemes: {A} general
  result}.
\newblock \bibinfo{journal}{{\em Econometrica\/}} \bibinfo{volume}{41},
  \bibinfo{number}{4} (\bibinfo{year}{1973}), \bibinfo{pages}{587--601}.
\newblock


\bibitem[\protect\citeauthoryear{Gibbard}{Gibbard}{1977}]%
        {Gibb77a}
\bibfield{author}{\bibinfo{person}{A. Gibbard}.}
  \bibinfo{year}{1977}\natexlab{}.
\newblock \showarticletitle{Manipulation of schemes that mix voting with
  chance}.
\newblock \bibinfo{journal}{{\em Econometrica\/}} \bibinfo{volume}{45},
  \bibinfo{number}{3} (\bibinfo{year}{1977}), \bibinfo{pages}{665--681}.
\newblock


\bibitem[\protect\citeauthoryear{Golowich, Narasimhan, and Parkes}{Golowich
  et~al\mbox{.}}{2018}]%
        {GNP18}
\bibfield{author}{\bibinfo{person}{N. Golowich}, \bibinfo{person}{H.
  Narasimhan}, {and} \bibinfo{person}{D.~C. Parkes}.}
  \bibinfo{year}{2018}\natexlab{}.
\newblock \showarticletitle{Deep Learning for Multi-Facility Location Mechanism
  Design}. In \bibinfo{booktitle}{{\em Proceedings of the 27th International
  Joint Conference on Artificial Intelligence (IJCAI)}}.
  \bibinfo{pages}{261--267}.
\newblock


\bibitem[\protect\citeauthoryear{Mei, Li, Ye, and Zhang}{Mei
  et~al\mbox{.}}{2016}]%
        {MLY+16}
\bibfield{author}{\bibinfo{person}{L. Mei}, \bibinfo{person}{M. Li},
  \bibinfo{person}{D. Ye}, {and} \bibinfo{person}{G. Zhang}.}
  \bibinfo{year}{2016}\natexlab{}.
\newblock \showarticletitle{Strategy-proof mechanism design for facility
  location games: {R}evisited}. In \bibinfo{booktitle}{{\em Proceedings of the
  15th International Joint Conference on Autonomous Agents and Multi-Agent
  Systems (AAMAS)}}.
\newblock


\bibitem[\protect\citeauthoryear{Moulin}{Moulin}{1980}]%
        {Moul80}
\bibfield{author}{\bibinfo{person}{H. Moulin}.}
  \bibinfo{year}{1980}\natexlab{}.
\newblock \showarticletitle{On Strategy-proofness and single peakedness}.
\newblock \bibinfo{journal}{{\em Public Choice\/}} \bibinfo{volume}{45},
  \bibinfo{number}{4} (\bibinfo{year}{1980}), \bibinfo{pages}{437--455}.
\newblock


\bibitem[\protect\citeauthoryear{Nisan and Ronen}{Nisan and Ronen}{2001}]%
        {NiRo01}
\bibfield{author}{\bibinfo{person}{N. Nisan} {and} \bibinfo{person}{A. Ronen}.}
  \bibinfo{year}{2001}\natexlab{}.
\newblock \showarticletitle{Algorithmic Mechanism Design}.
\newblock \bibinfo{journal}{{\em Games and Economic Behavior\/}}
  \bibinfo{volume}{35}, \bibinfo{number}{1--2} (\bibinfo{year}{2001}),
  \bibinfo{pages}{166--196}.
\newblock


\bibitem[\protect\citeauthoryear{P{\'a}l, Tardos, and Wexler}{P{\'a}l
  et~al\mbox{.}}{2001}]%
        {KPTW01}
\bibfield{author}{\bibinfo{person}{M. P{\'a}l}, \bibinfo{person}{{\'E}.
  Tardos}, {and} \bibinfo{person}{T. Wexler}.} \bibinfo{year}{2001}\natexlab{}.
\newblock \showarticletitle{Facility location with nonuniform hard capacities}.
  In \bibinfo{booktitle}{{\em Proceedings of the FOCS'01}}.
  \bibinfo{pages}{329--338}.
\newblock


\bibitem[\protect\citeauthoryear{Peremans, Peters, v.d. Stel, and
  Storcken}{Peremans et~al\mbox{.}}{1997}]%
        {PPS+97}
\bibfield{author}{\bibinfo{person}{W. Peremans}, \bibinfo{person}{H. Peters},
  \bibinfo{person}{H. v.d. Stel}, {and} \bibinfo{person}{T. Storcken}.}
  \bibinfo{year}{1997}\natexlab{}.
\newblock \showarticletitle{Strategy-proofness on {E}uclidean spaces}.
\newblock \bibinfo{journal}{{\em Social Choice and Welfare\/}}
  \bibinfo{volume}{14} (\bibinfo{year}{1997}), \bibinfo{pages}{379--401}.
\newblock


\bibitem[\protect\citeauthoryear{Procaccia and Tennenholtz}{Procaccia and
  Tennenholtz}{2013}]%
        {PrTe13}
\bibfield{author}{\bibinfo{person}{A.~D. Procaccia} {and} \bibinfo{person}{M.
  Tennenholtz}.} \bibinfo{year}{2013}\natexlab{}.
\newblock \showarticletitle{Approximate Mechanism Design Without Money}. In
  \bibinfo{booktitle}{{\em Proceedings of the 14th ACM Conference on Electronic
  Commerce (ACM-EC)}}. \bibinfo{publisher}{ACM Press}, \bibinfo{pages}{1--26}.
\newblock


\bibitem[\protect\citeauthoryear{Satterthwaite}{Satterthwaite}{1975}]%
        {Satt75}
\bibfield{author}{\bibinfo{person}{M.A. Satterthwaite}.}
  \bibinfo{year}{1975}\natexlab{}.
\newblock \showarticletitle{Strategy-proofness and {A}rrow's conditions:
  Existence and correspondence theorems for voting procedures and social
  welfare functions}.
\newblock \bibinfo{journal}{{\em Journal of Economic Theory\/}}
  \bibinfo{volume}{10} (\bibinfo{year}{1975}), \bibinfo{pages}{187--217}.
\newblock


\bibitem[\protect\citeauthoryear{Sui and Boutilier}{Sui and Boutilier}{2015}]%
        {SuBo15}
\bibfield{author}{\bibinfo{person}{X. Sui} {and} \bibinfo{person}{C.
  Boutilier}.} \bibinfo{year}{2015}\natexlab{}.
\newblock \showarticletitle{Approximately Strategy-proof mechanisms for
  (constrained) facility location}. In \bibinfo{booktitle}{{\em Proceedings of
  the 14th International Joint Conference on Autonomous Agents and Multi-Agent
  Systems (AAMAS)}}. \bibinfo{pages}{605--613}.
\newblock


\bibitem[\protect\citeauthoryear{Vygen}{Vygen}{2004}]%
        {Vygen05:Approximation}
\bibfield{author}{\bibinfo{person}{J. Vygen}.} \bibinfo{year}{2004}\natexlab{}.
\newblock \bibinfo{booktitle}{{\em Approximation algorithms facility location
  problems}}.
\newblock \bibinfo{publisher}{Lecture notes, Research Institute for Discrete
  Mathematics, University of Bonn, Germany}.
\newblock


\bibitem[\protect\citeauthoryear{Weymark}{Weymark}{2011}]%
        {Weym11}
\bibfield{author}{\bibinfo{person}{J.A. Weymark}.}
  \bibinfo{year}{2011}\natexlab{}.
\newblock \showarticletitle{A unified approach to strategy-proofness for
  single-peaked preferences}.
\newblock \bibinfo{journal}{{\em SERIEs\/}} \bibinfo{volume}{2},
  \bibinfo{number}{4} (\bibinfo{year}{2011}), \bibinfo{pages}{529--550}.
\newblock


\end{thebibliography}

\appendix

\section{Omitted Proofs}

\begin{proof}[Proof of Proposition~\ref{Proposition: Basic prop 1}]
Let \instance be an arbitrary instance, and consider the subgame $\Gamma_{\boldsymbol{x}}(s, k)$. We first show there always exists an equilibrium where agents in $N_k^*(\boldsymbol{x}, s)$ are served, and the others are not. To see this, suppose all $i\in N_k^*(\boldsymbol{x}, s)$ choose action $a_i=s$ and all other agents choose $a_i=\emptyset$. In this case, the $k$ agents in $N_k^*(\boldsymbol{x}, s)$ attain utility $1-d(s, x_i)\ge 0$, and all other agents attain utility zero. An agent in $i\in N_k^*(\boldsymbol{x}, s)$ can never strictly benefit from deviating to $a_i=0$, since this leads to utility zero.  An agent $j$ not in $N_k^*(\boldsymbol{x}, s)$ can never strictly benefit from deviating to $a_j=s$, since the tie-breaking rule ($\triangleright$) would lead to the agent not being served -- hence, attaining utility $-d(s,x_j)\le 0$. 

We now show that an agent's equilibrium utility is invariant across equilibria (when multiple equilibria exist). Let $i\in N$ be some agent, let $\sigma, \sigma'$ be two distinct equilibria of the subgame $\Gamma_{\boldsymbol{x}}(s,k)$. Denote agent $i$'s utility in each of these equilibria by $\bar{u}_i, \bar{u}_i'$, respectively. Note that $\bar{u}_i, \bar{u}_i'\in \{0,\ 1-d(s,x_i)\}$.

For the sake of a contradiction suppose that $\bar{u}_i\neq \bar{u}_i'$, notice that this necessarily implies that $d(s,x_i)<1$ and $k<n$. Without loss of generality assume $\bar{u}_i=0$ and $\bar{u}_i'=1-d(s,x_i)>0$. If $\sigma$ is an equilibrium it must be that agent $i$ is not served when choosing action $a_i=s$; that is, agent $i$ is not in the set of $k$-closest agents $N_k^*(\boldsymbol{x}, s)$. Now consider the equilibrium $\sigma'$ where agent $i$ is served. Given that the facility has capacity $k<n$, and agent $i$ is served despite $i\notin N_k^*(\boldsymbol{x}, s)$ there must be an agent $j\in N_k^*(\boldsymbol{x}, s)$ such that they choose action $a_j'=\emptyset$ (and are not served). In this case, agent $j$ attain utility zero in equilibrium $\sigma'$. If instead agent $j$ unilaterally deviated to the action $a_j''=s$ they would be served and attain utility $1-d(s, x_j)\ge 0$. Furthermore, $1-d(s, x_j)> 0$ since agent $j \in N_k^*(\boldsymbol{x}, s)$ and $i\notin N_k^*(\boldsymbol{x}, s)$ and so $d(s, x_j)\le d(s, x_i)<1$. We conclude that $\sigma'$ is not an equilibrium; that is, we have a contradiction.
\end{proof}

\begin{proof}[Proof of Proposition~\ref{Proposition: Basic prop 2}]
Let  \instance and $\langle \boldsymbol{x}, s', k\rangle$ be two instances such that for some agent $i$ $s<s'\le x_i$, or $x_i\le s'<s$. From Proposition 3.1, in the first instance we know that agent $i$ attains utility $1-d(s, x_i)$ whenever $i\in  N_k^*(\boldsymbol{x}, s)$ and otherwise attains utility zero. Similarly, in the second instance  $i$ attains utility $1-d(s', x_i)$ whenever $i\in  N_k^*(\boldsymbol{x}, s')$ and otherwise attains utility zero. The set $N_k^*(\boldsymbol{x}, s)$ is defined as the  $k$-closest agents with respect to $\triangleright$ to the facility, and $N_k^*(\boldsymbol{x}, s')$ is similarly defined with respect to the priority  $\triangleright'$.  If $s<s'\le x_i$ or $x_i\le s'<s$ then agent $i$'s priority under  $\triangleright'$ (weakly) improves compared to their priority under  $\triangleright$. Thus, $i\in N_k^*(\boldsymbol{x}, s)$  implies $i\in N_k^*(\boldsymbol{x}, s')$. We conclude that agent $i$ attains weakly higher utility in the $\langle \boldsymbol{x}, s', k\rangle$ instance.
\end{proof}

\end{document}